\documentclass[12pt,english]{article}
\usepackage[T1]{fontenc}
\usepackage[latin9]{inputenc}
\usepackage{geometry}
\usepackage{verbatim}
\usepackage{amsthm}
\usepackage{amsmath}
\usepackage{amssymb}
\usepackage{graphicx}
\usepackage{setspace}
\usepackage{esint}
\usepackage[authoryear]{natbib}
\usepackage{hyperref}
\usepackage{xcolor}

\hypersetup{
	colorlinks,
	linkcolor={red!50!black},
	citecolor={blue!50!black},
	urlcolor={blue!80!black}
}

\makeatletter

\newcommand{\mathcircumflex}[0]{\mbox{\^{}}}

\providecommand{\tabularnewline}{\\}

  \theoremstyle{definition}
  \newtheorem{defn}{\protect\definitionname}
  \theoremstyle{plain}
  \newtheorem{lem}{\protect\lemmaname}
  \theoremstyle{remark}
  \newtheorem{rem}{\protect\remarkname}
\theoremstyle{plain}
\newtheorem{thm}{\protect\theoremname}
  \theoremstyle{plain}
  \newtheorem{prop}{\protect\propositionname}

\usepackage{amsfonts}\usepackage{babel}
\providecommand{\U}[1]{\protect\rule{.1in}{.1in}}
 \usepackage{fancyheadings}

\makeatother

\usepackage{babel}
  \providecommand{\definitionname}{Definition}
  \providecommand{\lemmaname}{Lemma}
  \providecommand{\propositionname}{Proposition}
  \providecommand{\remarkname}{Remark}
\providecommand{\theoremname}{Theorem}

\begin{document}

\title{Equilibrium in Misspecified\\
Markov Decision Processes \thanks{We thank Vladimir Asriyan, Hector Chade, Xiaohong Chen, Emilio Espino,
Drew Fudenberg, Bruce Hansen, Philippe Jehiel, Jack Porter, Philippe
Rigollet, Tom Sargent, Iván Werning, and several seminar participants
for helpful comments. Esponda: Olin Business School, Washington University
in St. Louis, 1 Brookings Drive, Campus Box 1133, St. Louis, MO 63130,
iesponda@wustl.edu; Pouzo: Department of Economics, UC Berkeley, 530-1
Evans Hall \#3880, Berkeley, CA 94720, dpouzo@econ.berkeley.edu.} \bigskip{}
}

\author{%
\begin{tabular}{cc}
Ignacio Esponda~~~~~~~~  & ~~~~~~~~Demian Pouzo\tabularnewline
(WUSTL)~~~~~~~~ & ~~~~~~~~(UC Berkeley)\tabularnewline
\end{tabular}}
\maketitle
\begin{abstract}
We study Markov decision problems where the agent does not know the
transition probability function mapping current states and actions
to future states. The agent has a prior belief over a set of possible
transition functions and updates beliefs using Bayes' rule. We allow
her to be misspecified in the sense that the true transition probability
function is not in the support of her prior. This problem is relevant
in many economic settings but is usually not amenable to analysis
by the researcher. We make the problem tractable by studying asymptotic
behavior. We propose an equilibrium notion and provide conditions
under which it characterizes steady state behavior. In the special
case where the problem is static, equilibrium coincides with the single-agent
version of Berk-Nash equilibrium (\citealp{esponda2016berk}). We
also discuss subtle issues that arise exclusively in dynamic settings
due to the possibility of a negative value of experimentation.
\end{abstract}
\bigskip{}

\thispagestyle{empty}

\newpage{}

\tableofcontents\thispagestyle{empty}\newpage{}

\setcounter{page}{1}

\section{Introduction}

Early interest on studying the behavior of agents who hold misspecified
views of the world (e.g., \cite{Arrow-Green}, \cite{kirman75learning},
\cite{sobel1984non}, \cite{kagel1986winner}, \cite{nyarko1991learning},
\cite{sargent-book}) has recently been renewed by the work of \cite{piccione2003modeling},
\cite{jehiel2005analogy}, \cite{eyster2005cursed}, \cite{jehiel2008revisiting},
\cite{esponda2008behavioral}, Esponda and Pouzo (2012, 2016)\nocite{esponda2016berk}\nocite{esponda2012conditional},
\cite{eyster2013approach}, Spiegler (2013, 2016a, 2016b)\nocite{spiegler2013placebo}\nocite{spiegler2016abayesian}\nocite{spiegler2016bon},
\cite{heidhues2016unrealistic}, and \cite{fudenberg2016active}.
There are least two reasons for this interest. First, it is natural
for agents to be uncertain about their complex environment and to
represent this uncertainty with parsimonious parametric models that
are likely to be misspecified. Second, endowing agents with misspecified
models can explain how certain biases in behavior arise endogenously
as a function of the primitives.\footnote{We take the misspecified model as a primitive and assume that agents
learn and behave optimally given their model. In contrast, \cite{hansen2008robustness}
study optimal behavior of agents who have a preference for robustness
because they are aware of the possibility of model misspecification.}

The previous literature mostly focuses on problems that are intrinsically
``static'' in the sense that they can be viewed as repetitions of
static problems where the only link between periods arises because
the agent is learning the parameters of the model. Yet dynamic decision
problems, where an agent chooses an action that affects a state variable
(other than a belief), are ubiquitous in economics. The main goal
of this paper is to provide a tractable framework to study dynamic
settings where the agent learns with a possibly misspecified model.

We study a Markov Decision Process where a single agent chooses actions
at discrete time intervals. A transition probability function describes
how the agent's action and the current state affects next period's
state. The current payoff is a function of states and actions. We
assume that the agent is uncertain about the true transition probability
function and wants to maximize expected discounted payoff. She has
a prior belief over a set of possible transition functions, and her
model is possibly misspecified, meaning that we do not require the
true transition probability function to be in the support of her prior.
The agent uses Bayes' rule to update her belief after observing the
realized state.

To better illustrate the main question and results, consider a dynamic
savings problem with unknown returns, where $s$ is current income,
$x$ is the choice of savings, $\pi(s-x)$ is the payoff from current
consumption, and next period's income $s'$ is drawn from the distribution
$Q(\cdot\mid s,x)$. The agent, however, does not know the return
distribution $Q$. She has a parametric model representing the set
of possible return distributions $Q_{\theta}$ indexed by a parameter
$\theta\in\Theta$. The agent has a prior $\mu$ over $\Theta$, and
this belief is updated using Bayes' rule based on current income,
the savings decision, and the income realized next period, $\mu'=B(s,x,s',\mu)$,
where $B$ denotes the Bayesian operator and $\mu'$ is the posterior
belief. The agent is correctly specified if the support of her prior
includes the true return distribution $Q$ and is misspecified otherwise.
We represent this problem recursively via the following Bellman equation:
\begin{equation}
W(s,\mu)=\max_{x\in[0,s]}\pi(s-x)+\delta\int\int W(s',\mu')Q_{\theta}(ds'\mid s,x)\mu(d\theta),\label{eq:savings_Bellman}
\end{equation}

The solution to this Bellman equation determines the evolution of
states, actions, and beliefs. A large computational literature provides
algorithms that agents and researchers can use to approximate the
solution to problems such as (\ref{eq:savings_Bellman}), where a
belief is part of the state variable; see \cite{powell2007approximate}
for a textbook treatment.\footnote{Of course, we do not expect less sophisticated agents to apply these
numerical methods. But, following the standard view in the literature,
the dynamic programming approach is still a useful tool for the researcher
to model the behavior of an agent facing intertemporal tradeoffs.} The issue for economists, however, is that these numerical methods
do not usually allow us to make general predictions about behavior.

We propose to circumvent this problem by instead characterizing the
agent's steady state behavior and beliefs. The main question that
we ask is whether we can replace a dynamic programming problem with
learning, such as (\ref{eq:savings_Bellman}), by a problem where
beliefs are not being updated, such as 
\begin{equation}
V(s)=\max_{x\in[0,s]}\pi(s-x)+\int V(s')\bar{Q}_{\mu^{*}}(ds'\mid s,x),\label{eq:savings_Bellman2}
\end{equation}
where $\mu^{*}$ is the agent's equilibrium or steady-state belief
over $\Theta$ and $\bar{Q}_{\mu^{*}}=\int_{\Theta}Q_{\theta}\mu^{*}(d\theta)$
is the corresponding subjective transition probability function. We
refer to this problem as a \emph{Markov Decision Process} (MDP) with
transition probability function $\bar{Q}_{\mu^{*}}$. The main advantage
of this approach is that, provided that we can characterize the equilibrium
belief $\mu^{*}$, it obviates the need to include beliefs in the
state space, thus making the problem much more amenable to analysis.
This focus on \emph{equilibrium} behavior is indeed a distinguishing
feature of economics.

We begin by defining a notion of equilibrium to capture the steady
state behavior and belief of an agent who does not know the true transition
probability function. We call this notion a Berk-Nash equilibrium
because, in the special case where the environment is static, it collapses
to the single-agent version of Berk-Nash equilibrium, a concept introduced
by \cite{esponda2016berk} to characterize steady state behavior
in static environments with misspecified agents. A strategy in an
MDP is a mapping from states to actions; recall that beliefs are not
included in the state for an MDP. For a given strategy and true transition
probability function, the stochastic process for states and actions
in an MDP is a Markov chain and has a corresponding stationary distribution
that can be interpreted as the steady-state distribution over outcomes.
A strategy and corresponding stationary distribution is a Berk-Nash
equilibrium if there exists a belief $\mu^{*}$ over the parameter
space such that: (i) the strategy is optimal for an MDP with transition
probability function $\bar{Q}_{\mu^{*}}$, and (ii) $\mu^{*}$ puts
probability one on the set of parameter values that yield transition
probability functions that are ``closest'' to the true transition
probability function. The notion of \textquotedblleft closest\textquotedblright{}
is given by a weighted version of the Kullback-Leibler divergence
that depends on the equilibrium stationary distribution.

We use the framework to revisit three classic examples. These examples
illustrate how our framework makes dynamic environments with uncertainty
amenable to analysis and expands the scope of the classical dynamic
programming approach. First, we consider the classic problem of a
monopolist with unknown demand function. We assume that demand is
dynamic, so that a sale in the current period affects the likelihood
of a sale the next period. The monopolist, however, has a misspecified
model and believes that demand is not dynamic. We show that a monopolist
who thinks demand is not dynamic does not necessarily set higher prices.

The second illustrative example is a search model where a worker does
not realize that she gets fired with higher probability in times in
which it is actually harder to find another job. We show that she
becomes pessimistic about the chances of finding a new job and sub-optimally
accepts wage offers that are too low.

The final example is a stochastic growth model along the lines of
the problem represented by (\ref{eq:savings_Bellman}). The agent
determines how much of her income to invest every period, which determines,
together with an unknown productivity process, next period's income.
We assume that there are correlated shocks to both the agent's utility
and productivity, but the agent believes these shocks to be independent.
If the shocks are positively correlated, the misspecified agent invests
more of her income when productivity is low. She ends up underestimating
productivity and, therefore, underinvesting in equilibrium.

We then turn to providing a foundation for Berk-Nash equilibrium by
studying the limiting behavior of a Bayesian agent who takes actions
and updates her beliefs about the transition probability function
every period. We ask if an equilibrium approach is appropriate in
this environment, i.e., ``Is it possible to characterize the steady
state behavior of a Bayesian agent by reference to a simpler MDP in
which the agent has fixed (though possibly incorrect) beliefs about
the transition probability function?''

The answer is yes if the agent is sufficiently impatient. But, if
the agent is sufficiently patient, some subtle issues arise in the
dynamic setting that lead to a more nuanced answer: The answer is
yes provided that we restrict attention to steady states with a property
we call exhaustive learning. Under exhaustive learning, the agent
perceives that she has nothing else to learn in steady state. In the
context of the previous example, this condition guarantees that optimal
actions in problem (\ref{eq:savings_Bellman}) are also optimal in
problem (\ref{eq:savings_Bellman2}). Without exhaustive learning,
an action may be optimal in problem (\ref{eq:savings_Bellman2}) because
the agent is not updating her beliefs. But the same action could be
suboptimal if she were to update beliefs because, as we show in this
paper, the value of experimentation can be negative in dynamic settings.
This situation is not possible in static settings because the value
function is only a function of beliefs and its convexity and the martingale
property of Bayesian beliefs imply that the value of experimentation
is always nonnegative.

The notion of exhaustive learning motivates a natural refinement of
Berk-Nash equilibrium in dynamic settings. This refinement, however,
still allows beliefs to be incorrect due to lack of experimentation,
which is a hallmark of the bandit (e.g., \cite{rothschild1974two},
\cite{mclennan1984price}, \cite{easley1988controlling}) and self-confirming
equilibrium (e.g., \cite{battigalli1987compartamento}, \cite{fudenberg1993self},
\cite{dekel2004learning}, \cite{fershtman2012dynamic}) literatures.
Following \cite{selten1975reexamination}, we define a further refinement,
\emph{perfect} Berk-Nash equilibrium, to characterize behavior that
is robust to experimentation, and provide conditions for its existence.

Our asymptotic characterization of beliefs and actions contributes
to the literature that studies asymptotic beliefs and/or behavior
under Bayesian learning. Table 1 categorizes some of the more relevant
papers in connection to our work. The table on the left includes papers
where the agent learns from data that is exogenous in the sense that
she does not affect the stochastic properties of the data. This topic
has mostly been tackled by statisticians for both correctly-specified
and misspecified models and for both i.i.d. and non-i.i.d. data. The
table on the right includes papers where the agent learns from data
that is endogenous in the sense that it is driven by the agent's actions,
a topic that has been studied by economists mostly in static settings.
By static we mean that the problem reduces to a static optimization
problem if stripped of the learning dynamics.\footnote{Formally, we say a problem is static if, for a fixed strategy and
belief over the transition probability function, outcomes (states
and actions) are independent across time.}

\begin{table}
\advance\leftskip-1.75cm\includegraphics[scale=0.62]{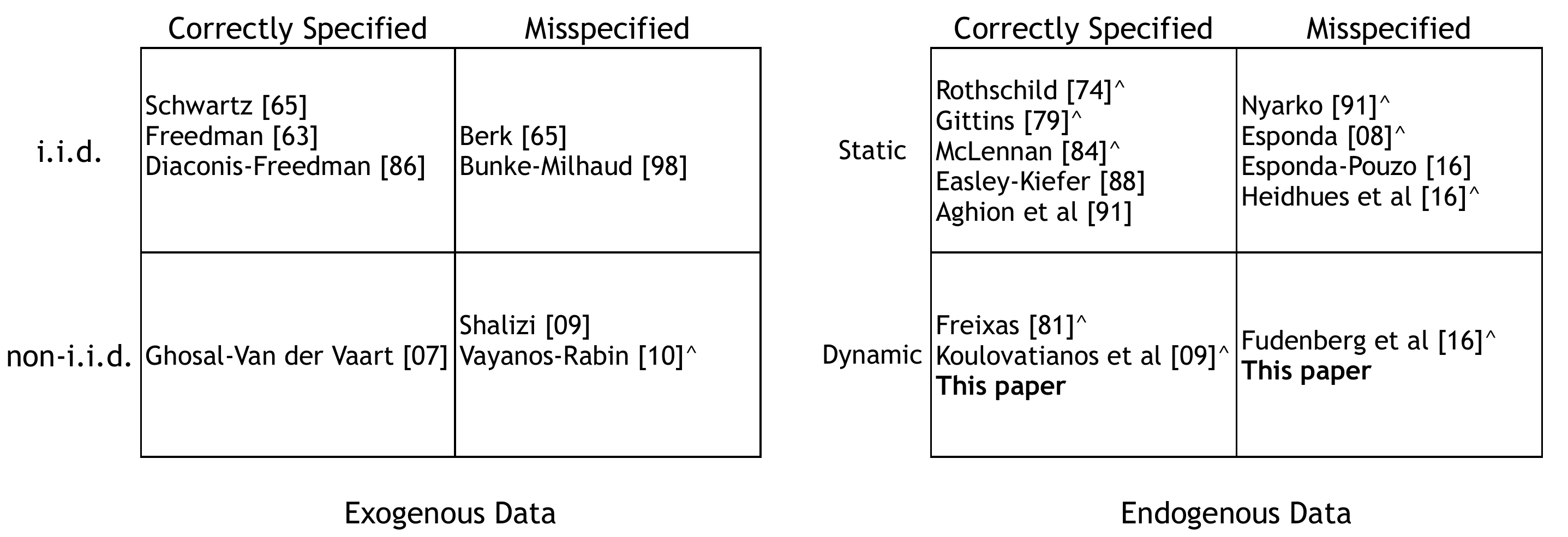}

\protect\caption{\label{tab:Literature}Literature on Bayesian Learning}
\end{table}

Table 1 also differentiates between two complementary approaches to
studying asymptotic beliefs and/or behavior. The first approach is
to focus on specific settings and provide a complete characterization
of asymptotic actions and beliefs, including convergence results;
these papers are marked with a superscript $\mathcircumflex$ in Table
1. Some papers pursue this approach in dynamic and correctly specified
stochastic growth models (e.g., \cite{freixas1981optimal}, \cite{koulovatianos2009optimal}).
In static misspecified settings, \cite{nyarko1991learning}, \cite{esponda2008behavioral},
and \cite{heidhues2016unrealistic} study passive learning problems
where there is no experimentation motive. \cite{fudenberg2016active}
is the only paper that provides a complete characterization in a dynamic
decision problem with active learning.\footnote{Under active learning, different actions convey different amount of
information and a non-myopic agent takes the exploitation vs. experimentation
tradeoff into account. There can be passive or active learning in
both static and dynamic settings. },\footnote{The environment in \cite{fudenberg2016active} is dynamic because
the agent controls the drift of a Brownian motion, even though the
only relevant state variable for optimality ends up being the agent's
belief.} The second approach, which we follow in this paper and we followed
earlier for the static case (\citealp{esponda2016berk}) is to study
general settings and focus on characterizing the set of steady states.\footnote{In macroeconomics there are several models where agents make forecasts
using statistical models that are misspecified (e.g., \cite{Evans-book}
Ch. 13, \cite{sargent-book} Ch. 6).}%

The paper is also related to the literature which provides learning
foundations for equilibrium concepts, such as Nash or self-confirming
equilibrium (see \cite{fudenberg1998theory} for a survey). In contrast
to this literature, we consider Markov decision problems and allow
for misspecified models. Particular types of misspecifications have
been studied in extensive form games. \cite{jehiel1995limited} considers
the class of repeated alternating-move games and assumes that players
only forecast a limited number of time periods into the future; see
\cite{jehiel1998learning} for a learning foundation. We share the
feature that the learning process takes place within the play of the
game and that beliefs are those that provide the best fit given the
data.\footnote{\cite{jehiel2007valuation} consider the general class of extensive
form games with perfect information and assume that players simplify
the game by partitioning the nodes into similarity classes.}

The framework and equilibrium notion are presented in Sections 2 and
3. In Section 4, we work through several examples. We provide a foundation
for equilibrium in Section 5 and study equilibrium refinements in
Section 6.

\section{\label{sec:MDP}Markov Decision Processes}

We begin by describing the environment faced by the agent.

\medskip{}

\begin{defn}
A \textbf{Markov Decision Process} (MDP) is a tuple $\left\langle \mathbb{S},\mathbb{X},\Gamma,q_{0},Q,\pi,\delta\right\rangle $
where\end{defn}
\begin{itemize}
\item $\mathbb{S}$ is a nonempty and finite set of states
\item $\mathbb{X}$ is a nonempty and finite set of actions
\item $\Gamma:\mathbb{S}\rightarrow2^{\mathbb{X}}$ is a non-empty constraint
correspondence
\item $q_{0}\in\Delta(\mathbb{S})$ is a probability distribution on the
initial state
\item $Q:Gr(\Gamma)\rightarrow\Delta(\mathbb{S})$ is a transition probability
function\footnote{For a correspondence $\Gamma:\mathbb{S}\rightarrow2^{\mathbb{X}}$,
its graph is defined by $Gr(\Gamma)\equiv\{(s,x)\in\mathbb{S}\times\mathbb{X}:x\in\Gamma(s)\}.$}
\item $\pi:Gr(\Gamma)\times\mathbb{S}\rightarrow\mathbb{R}$ is a per-period
payoff function
\item $\delta\in[0,1)$ is a discount factor
\end{itemize}
\medskip{}

We sometimes use MDP($Q$) to denote an MDP with transition probability
function $Q$ and exclude the remaining primitives.

The timing is as follows. At the beginning of each period $t=0,1,2,...$,
the agent observes state $s_{t}\in\mathbb{S}$ and chooses a feasible
action $x_{t}\in\Gamma(s_{t})\subset\mathbb{X}$. Then a new state
$s_{t+1}$ is drawn according to the probability distribution $Q(\cdot\mid s_{t},x_{t})$
and the agent receives payoff $\pi(s_{t},x_{t},s_{t+1})$ in period
$t$. The initial state $s_{0}$ is drawn according to the probability
distribution $q_{0}$.

The agent facing an MDP chooses a policy rule that specifies at each
point in time a (possibly random) action as a function of the history
of states and actions observed up to that point. As usual, the objective
of the agent is to choose a feasible policy rule to maximize expected
discounted utility, $\sum_{t=0}^{\infty}\delta^{t}\pi(s_{t},x_{t},s_{t+1})$.

By the Principle of Optimality, the agent's problem can be cast recursively
as 
\begin{equation}
V_{Q}(s)=\max_{x\in\Gamma(s)}\int_{\mathbb{S}}\left\{ \pi(s,x,s')+\delta V_{Q}(s')\right\} Q(ds'|s,x)\label{eq:bellman}
\end{equation}
where $V_{Q}:\mathbb{S}\rightarrow\mathbb{R}$ is the (unique) solution
to the Bellman equation (\ref{eq:bellman}).

\medskip{}

\begin{defn}
A \textbf{strategy} $\sigma$ is a distribution over actions given
states, $\sigma:\mathbb{S}\rightarrow\Delta(\mathbb{X})$, that satisfies
$\sigma(s)\in\Gamma(s)$ for all $s$.
\end{defn}
\medskip{}

Let $\Sigma$ denote the space of all strategies and let $\sigma(x\mid s)$
denote the probability that the agent chooses $x$ when the state
is $s$.\footnote{A standard result is the existence of a \emph{deterministic} optimal
strategy. Nevertheless, allowing for randomization will be important
in the case where the transition probability function is uncertain.}

\medskip{}

\begin{defn}
A strategy $\sigma\in\Sigma$ is \textbf{optimal} for an MDP($Q$)
if, for all $s\in\mathbb{S}$ and all $x\in\mathbb{X}$ such that
$\sigma(x\mid s)>0$, 
\[
x\in\arg\max_{\hat{x}\in\Gamma(s)}\int_{\mathbb{S}}\left\{ \pi(s,\hat{x},s')+\delta V_{Q}(s')\right\} Q(ds'|s,\hat{x}).
\]

\end{defn}
\medskip{}

Let $\Sigma(Q)$ be the set of all strategies that are optimal for
an MDP($Q$).

\medskip{}

\begin{lem}
\label{Lemma:Sigma(Q)}(i) There is a unique solution $V_{Q}$ to
the Bellman equation in (\ref{eq:bellman}), and it is continuous
in $Q$ for all $s\in\mathbb{S}$; (ii) The correspondence of optimal
strategies $Q\mapsto\Sigma(Q)$ is non-empty, compact-valued, convex-valued,
and upper hemicontinuous.\end{lem}
\begin{proof}
The proof is standard and relegated to the Online Appendix.
\end{proof}
\medskip{}

A strategy determines the transitions in the space of states and actions
and, consequently, the set of stationary distributions over states
and actions. For any strategy $\sigma$ and transition probability
function $Q$, define a \textbf{transition kernel} $M_{\sigma,Q}:Gr(\Gamma)\rightarrow\Delta\left(Gr(\Gamma)\right)$
by letting 
\begin{equation}
M_{\sigma,Q}(s',x'\mid s,x)=\sigma(x'\mid s')Q(s'\mid s,x)\label{eq:transition_kernel-1}
\end{equation}
for all $(s,x),(s',x')\in Gr(\Gamma)$. The transition kernel $M_{\sigma,Q}$
is the transition probability function over $Gr(\Gamma)$ given strategy
$\sigma$ and transition probability function $Q$.

For any $m\in\Delta(Gr(\Gamma))$, let $M_{\sigma,Q}[m]\in\Delta(Gr(\Gamma))$
denote the probability measure 
\[
\sum_{(s,x)\in Gr(\Gamma)}M_{\sigma,Q}(\cdot,\cdot\mid s,x)m(s,x).
\]

\begin{defn}
A distribution $m\in\Delta(Gr(\Gamma))$ is a \textbf{stationary (or
invariant) distribution} given $(\sigma,Q)$ if $m=M_{\sigma,Q}[m]$.
\end{defn}
\medskip{}

A stationary distribution represents the steady-state distribution
over outcomes (i.e, states and actions) when the agent follows a given
strategy. Let $I_{Q}(\sigma)\equiv\{m\in\Delta(Gr(\Gamma))\mid m=M_{\sigma,Q}[m]\}$
denote the set of stationary distributions given $(\sigma,Q)$.\medskip{}

\begin{lem}
\label{Lemma:I(sigma)}The correspondence of stationary distributions
$\sigma\mapsto I_{Q}(\sigma)$ is non-empty, compact-valued, convex-valued,
and upper hemicontinuous.\end{lem}
\begin{proof}
See the Appendix.
\end{proof}

\section{\label{sec:BMDP}Subjective Markov Decision Processes}

Our main objective is to study the behavior of an agent who faces
an MDP but is uncertain about the transition probability function.
We begin by introducing a new object to model the problem with uncertainty,
which we call the \emph{Subjective} \emph{Markov decision process}
(SMDP). We then define the notion of a Berk-Nash equilibrium of an
SMDP.

\subsection{Setup}
\begin{defn}
\label{def:BMDP}A \textbf{Subjective Markov Decision Process} (SMDP)
is an MDP, $\left\langle \mathbb{S},\mathbb{X},\Gamma,q_{0},Q,\pi,\delta\right\rangle $,
and a nonempty family of transition probability functions, $\mathcal{Q}_{\Theta}=\{Q_{\theta}:\theta\in\Theta\}$,
where each transition probability function $Q_{\theta}:Gr(\Gamma)\rightarrow\Delta(\mathbb{S})$
is indexed by a parameter $\theta\in\Theta$.
\end{defn}
\medskip{}

We interpret the set $\mathcal{Q}_{\Theta}$ as the different transition
probability functions (or models of the world) that the agent considers
possible. We sometimes use SMDP($Q,\mathcal{Q}_{\Theta}$) to denote
an SMDP with true transition probability function $Q$ and a family
of transition probability functions $\mathcal{Q}_{\Theta}$.\medskip{}

\begin{defn}
\label{Def:regular}A \textbf{Regular Subjective Markov Decision Process}
(regular-SMDP) is an SMDP that satisfies the following conditions\end{defn}
\begin{itemize}
\item $\Theta$ is a compact subset of an Euclidean space.
\item $Q_{\theta}(s'\mid s,x)$ is continuous as a function of $\theta\in\Theta$
for all $(s',s,x)\in\mathbb{S}\times Gr(\Gamma)$.
\item There is a dense set $\hat{\Theta}\subseteq\Theta$ such that, for
all $\theta\in\hat{\Theta}$, $Q_{\theta}(s'\mid s,x)>0$ for all
$(s',s,x)\in\mathbb{S}\times Gr(\Gamma)$ such that $Q(s'\mid s,x)>0$.%

\end{itemize}
\medskip{}

The first two conditions in Definition \ref{Def:regular} place parametric
and continuity assumptions on the subjective models.\footnote{Without the assumption of a finite-dimensional parameter space, Bayesian
updating need not converge to the truth for most priors and parameter
values even in correctly specified statistical settings (\cite{freedman1963asymptotic},
\cite{diaconis1986consistency}). Note that the parametric assumption
is only a restriction if the set of states or actions is nonfinite,
a case we consider in some of the examples.} The last condition plays two roles. First, it rules out a stark form
of misspecification by guaranteeing that there exists at least one
parameter value that can rationalize every feasible observation. Second,
it implies that the correspondence of parameters that are a closest
fit to the true model is upper hemicontinuous. \cite{esponda2016berk}
provide a simple (non-dynamic) example where this assumption does
not hold and equilibrium fails to exist.

\subsection{Equilibrium}

The goal of this section is to define the notion of Berk-Nash equilibrium
of an SMDP. The next definition is used to place constraints on the
belief $\mu\in\Delta(\Theta)$ that the agent may hold if $m$ is
the stationary distribution over outcomes.

\medskip{}

\begin{defn}
\label{def:wKLD}The \textbf{weighted Kullback-Leibler divergence}
(wKLD) is a mapping $K_{Q}\colon\Delta(Gr(\Gamma))\times\Theta\rightarrow\mathbb{\bar{R}}_{+}$
such that for any $m\in\Delta(Gr(\Gamma))$ and $\theta\in\Theta$,
\[
K_{Q}(m,\theta)=\sum_{(s,x)\in Gr(\Gamma)}E_{Q(\cdot|s,x)}\left[\ln\left(\frac{Q(S'|s,x)}{Q_{\theta}(S'|s,x)}\right)\right]m(s,x).
\]

The \textbf{set of closest parameter values given }$m\in\Delta(Gr(\Gamma))$
is the set 
\[
\Theta_{Q}(m)\equiv\arg\min_{\theta\in\Theta}K_{Q}(m,\theta).
\]

\end{defn}
\medskip{}

The set $\Theta_{Q}(m)$ contains the parameter values constitute
the best fit with the true transition probability function $Q$ when
outcomes are drawn from the distribution $m$. \medskip{}

\begin{lem}
\label{Lemma:Theta(m)}(i) For every $m\in\Delta(Gr(\Gamma))$ and
$\theta\in\Theta$, $K_{Q}(m,\theta)\geq0$, with equality holding
if and only if $Q_{\theta}(\cdot\mid s,x)=Q(\cdot\mid s,x)$ for all
$(s,x)$ such that $m(s,x)>0$. (ii) For any regular SMDP($Q,\mathcal{Q}_{\Theta}$),
$m\mapsto\Theta_{Q}(m)$ is non-empty, compact valued, and upper hemicontinuous.\end{lem}
\begin{proof}
See the Appendix.
\end{proof}
\medskip{}

We now define equilibrium.\medskip{}

\begin{defn}
\label{Def:Berk-Nash}A strategy and probability distribution $(\sigma,m)\in\Sigma\times\Delta(Gr(\Gamma))$
is a \textbf{Berk-Nash} \textbf{equilibrium} of the SMDP($Q,\mathcal{Q}_{\Theta}$)
if there exists a belief $\mu\in\Delta(\Theta)$ such that

(i) $\sigma$ is an optimal strategy for the MDP($\bar{Q}_{\mu}$),
where $\bar{Q}_{\mu}=\int_{\Theta}Q_{\theta}\mu(d\theta)$,

(ii) $\mu\in\Delta(\Theta_{Q}(m))$, and

(iii) $m\in I_{Q}(\sigma)$.
\end{defn}
\medskip{}

Condition (i) in the definition of Berk-Nash equilibrium requires
$\sigma$ to be an optimal strategy in the MDP where the transition
probability function is $\int_{\Theta}Q_{\theta}\mu(d\theta)$. Condition
(ii) requires that the agent only puts positive probability on the
set of closest parameter values given $m$, $\Theta_{Q}(m)$. Finally,
condition (iii) requires $m$ to be a stationary distribution given
$(\sigma,Q)$.\medskip{}

\begin{rem}
In Section \ref{sec:Equilibrium-foundation}, we interpret the set
of equilibria as the set of steady states of a learning environment
where the agent is uncertain about $Q$. The main advantage of the
equilibrium approach is that it allows us to replace a difficult learning
problem with a simpler MDP with a fixed transition probability function.
The cost of this approach is that it can only be used to characterize
asymptotic behavior, as opposed to the actual dynamics starting from
the initial distribution over states, $q_{0}\in\Delta(\mathbb{S})$.
This explains why $q_{0}$ does not enter the definition of equilibrium,
and why a mapping between $q_{0}$ and the set of corresponding equilibria
cannot be provided in general.
\end{rem}
\smallskip{}

\begin{rem}
In the special case of a static environment, Definition \ref{Def:Berk-Nash}
reduces to Esponda and Pouzo's (2016) definition of Berk-Nash equilibrium
for a single agent. In the dynamic environment, outcomes follow a
Markov process and we must keep track not only of strategies but also
of the corresponding stationary distribution over outcomes.

\medskip{}

\end{rem}
The next result establishes existence of equilibrium in any regular
SMDP.\medskip{}

\begin{thm}
\label{The:Existence}For any regular SMDP, there exists a Berk-Nash
equilibrium.\end{thm}
\begin{proof}
See the Appendix.
\end{proof}
\medskip{}

The standard approach to proving existence begins by defining a ``best
response correspondence'' in the space of strategies. This approach
does not work here because the possible non-uniqueness of beliefs
implies that the correspondence may not be convex valued. The trick
we employ is to define equilibrium via a correspondence on the space
of strategies, stationary distributions, and beliefs, and then use
Lemmas \ref{Lemma:Sigma(Q)}, \ref{Lemma:I(sigma)} and \ref{Lemma:Theta(m)}
to show that this correspondence satisfies the assumptions of a generalized
version of Kakutani's fixed point theorem.\footnote{\cite{esponda2016berk} rely on perturbations to show existence of
equilibrium in a static setting. In contrast, our approach does not
require the use of perturbations.}

\subsection{Correctly specified and identified SMDPs}

An SMDP is correctly specified if the set of subjective models contains
the true model. \medskip{}

\begin{defn}
An SMDP($Q,\mathcal{Q}_{\Theta}$) is \textbf{correctly specified}
if $Q\in\mathcal{Q}_{\Theta}$; otherwise, it is misspecified.
\end{defn}
\medskip{}

In decision problems, data is endogenous and so, following \cite{esponda2016berk},
it is natural to consider two notions of identification: weak and
strong identification. These definitions distinguish between outcomes
on and off the equilibrium path. In a dynamic environment, the right
object to describe what happens on and off the equilibrium path is
not the strategy but rather the stationary distribution over outcomes
$m$. \medskip{}

\begin{defn}
An SMDP is \textbf{weakly identified} \textbf{given} $\mathbf{m}\in\Delta(Gr(\Gamma))$
if $\theta,\theta'\in\Theta_{Q}(m)$ implies that $Q_{\theta}(\cdot\mid s,x)=Q_{\theta'}(\cdot\mid s,x)$
for all $(s,x)\in Gr(\Gamma)$ such that $m(s,x)>0$; if the condition
is satisfied for all $(s,x)\in Gr(\Gamma)$, we say that the \textbf{SMDP
is strongly identified given }$\mathbf{m}$. An SMDP is weakly (strongly)
identified if it is weakly (strongly) identified for all $m\in\Delta(Gr(\Gamma))$.
\end{defn}
\medskip{}

Weak identification implies that, for any equilibrium distribution
$m$, the agent has a unique belief along the equilibrium path, i.e.,
for states and actions that occur with positive probability. It is
a condition that turns out to be important for proving the existence
of equilibria that are robust to experimentation (see Section \ref{sec:Equilibrium-refinements})
and is always satisfied in correctly specified SMDPs.\footnote{The following is an example where weak identification fails. Suppose
an unbiased coin is tossed every period, but the agent believes that
the coin comes up heads with probability 1/4 or 3/4, but not 1/2.
Then both 1/4 and 3/4 minimize the Kullback-Leibler divergence, but
they imply different distributions over outcomes. Relatedly, Berk
(1966) shows that beliefs do not converge.} Strong identification strengthens the condition by requiring that
beliefs are unique also off the equilibrium path.

\medskip{}

\begin{prop}
Consider a correctly specified and strongly identified SMDP with corresponding
MDP($Q$). A strategy and probability distribution $(\sigma,m)\in\Sigma\times\Delta(Gr(\Gamma))$
is a Berk-Nash equilibrium of the SMDP if and only if $\sigma$ is
optimal given MDP($Q$) and $m$ is a stationary distribution given
$\sigma$.\end{prop}
\begin{proof}
\emph{Only if}: Suppose $(\sigma,m)$ is a Berk-Nash equilibrium.
Then there exists $\mu$ such that $\sigma$ is optimal given MDP($\bar{Q}_{\mu}$),
$\mu\in\Delta(\Theta(m))$, and $m\in I_{Q}(\sigma)$. Because the
SMDP is correctly specified, there exists $\theta^{*}$ such that
$Q_{\theta^{*}}=Q$ and, therefore, by Lemma \ref{Lemma:Theta(m)}(i),
$\theta^{*}\in\Delta(\Theta(m))$. Then, by strong identification,
any $\hat{\theta}\in\Theta(m)$ satisfies $Q_{\hat{\theta}}=Q_{\theta^{*}}=Q$,
implying that $\sigma$ is also optimal given MDP($Q$). \emph{If}:
Let $m\in I_{Q}(\sigma)$, where $\sigma$ is optimal given MDP($Q$).
Because the SMDP is correctly specified, there exists $\theta^{*}$
such that $Q_{\theta^{*}}=Q$ and, therefore, by Lemma \ref{Lemma:Theta(m)}(i),
$\theta^{*}\in\Delta(\Theta(m))$. Thus, $\sigma$ is also optimal
given $Q_{\theta^{*}}$, implying that $(\sigma,m)$ is a Berk-Nash
equilibrium.
\end{proof}
\medskip{}

Proposition 1 says that, in environments where the agent is uncertain
about the transition probability function but her subjective model
is both correctly specified and strongly identified, then Berk-Nash
equilibrium corresponds to the solution of the MDP under correct beliefs
about the transition probability function. If one drops the assumption
that the SMDP is strongly identified, then the ``if'' part of the
proposition continues to hold but the ``only if'' condition does
not hold. In other words, there may be Berk-Nash equilibria of correctly-specified
SMDPs in which the agent has incorrect beliefs off the equilibrium
path. This feature of equilibrium is analogous to the main ideas of
the bandit and self-confirming equilibrium literatures.

\section{Examples\label{sec:Examples}}

We use three classic examples to illustrate how easy it is to use
our framework to expand the scope of the classical dynamic programming
approach.

\subsection{\label{sub:Monopolist}Monopolist with unknown dynamic demand}

The problem of a monopolist facing an unknown, \emph{static} demand
function was first studied by \cite{rothschild1974two} and \cite{nyarko1991learning}
in correctly and misspecified settings, respectively. In the following
example, the monopolist faces a dynamic demand function but incorrectly
believes that demand is static.

\textbf{MDP}: In each period $t$, a monopolist chooses price $x_{t}\in\mathbb{X}=\{L,H\}$,
where $0<L<H$. It then sells $s_{t+1}\in\mathbb{S}=\{0,1\}$ units
at zero cost and obtains profit $\pi(x_{t},s_{t+1})=x_{t}s_{t+1}$.
The probability that $s_{t+1}=1$ is $q_{sx}\equiv Q(1\mid s_{t}=s,x_{t}=x)$,
where $0<q_{sx}<1$ for all $(s,x)\in Gr(\Gamma)=\mathbb{S}\times\mathbb{X}$.\footnote{The set of feasible actions is independent of the state, i.e., $\Gamma(s)=\mathbb{X}$
for all $s\in\mathbb{S}$.} The monopolist wants to maximize expected discounted profits, with
discount factor $\delta\in[0,1)$.

Demand is dynamic in the sense that a sale yesterday increases the
probability of a sale today: $q_{1x}>q_{0x}$ for all $x\in\mathbb{X}$.
Moreover, a higher price reduces the probability of a sale: $q_{sL}>q_{sH}$
for all $s\in\mathbb{S}$. Finally, for concreteness, we assume that
\begin{equation}
\frac{q_{1L}}{q_{1H}}<\frac{H}{L}<\frac{q_{0L}}{q_{0H}}.\label{eq:assumption_priceratio}
\end{equation}
Expression (\ref{eq:assumption_priceratio}) implies that current-period
profits are maximized by choosing price $L$ if there was no sale
last period and price $H$ otherwise (i.e., $Lq_{0L}>Hq_{0H}$ and
$Hq_{1H}>Lq_{1L}$). Thus, the optimal strategy of a myopic monopolist
(i.e., $\delta=0$) who knows the primitives is $\sigma(H\mid0)=0$
and $\sigma(H\mid1)=1$. If, however, the monopolist is sufficiently
patient, it is optimal to always choose price $L$.\footnote{\label{fn:C_delta}Formally, there exists $C_{\delta}\in[q_{1L}/q_{1H},q_{0L}/q_{0H}]$,
where $C_{0}=q_{1L}/q_{1H}$ and $\delta\mapsto C_{\delta}$ is increasing,
such that, if $H/L<C_{\delta}$, the optimal strategy is $\sigma(H\mid0)=\sigma(H\mid1)=0$.}

\textbf{SMDP}. The monopolist does not know $Q$ and believes, incorrectly,
that demand is not dynamic. Formally, $\mathcal{Q}_{\Theta}=\{Q_{\theta}:\theta\in\Theta\}$,
where $\Theta=[0,1]^{2}$ and, for all $\theta=(\theta_{L},\theta_{H})\in\Theta$,
$Q_{\theta}(1\mid s,L)=\theta_{L}$ and $Q_{\theta}(1\mid s,H)=\theta_{H}$
for all $s\in\mathbb{S}$. In particular, $\theta_{x}$ is the probability
that a sale occurs given price $x\in\{L,H\}$, and the agent believes
that it does not depend on $s$. Note that this SMDP is regular. For
simplicity, we restrict attention to equilibria in which the monopolist
does not condition on last period's state, and denote a strategy by
$\sigma_{H}$, the probability that price $H$ is chosen.

\textbf{Equilibrium}. %
\emph{Optimality}. Because the monopolist believes that demand is
static, the optimal strategy is to choose the price that maximizes
current period's profit. Let 
\[
\Delta(\theta)\equiv H\theta_{H}-L\theta_{L}
\]
denote the perceived expected payoff difference of choosing $H$ vs.
$L$ under the belief that the parameter value is $\theta=(\theta_{L},\theta_{H})$
with probability 1. If $\Delta(\theta)>0$, $\sigma_{H}=1$ is the
unique optimal strategy; if $\Delta(\theta)<0$, $\sigma_{H}=0$ is
the unique optimal strategy; and if $\Delta(\theta)=0$, any $\sigma_{H}\in[0,1]$
is optimal.

\emph{Beliefs}. For any $m\in\Delta(\mathbb{S}\times\mathbb{X})$,
the wKLD simplifies to 
\[
K_{Q}(m,\theta)=\sum_{x\in\{L,H\}}m_{\mathbb{X}}(x)\left\{ \bar{s}_{x}(m)\ln\theta_{x}+(1-\bar{s}_{x}(m))\ln(1-\theta_{x})\right\} +Const,
\]
where $\bar{s}_{x}(m)=m_{\mathbb{S}\mid\mathbb{X}}(0\mid x)q_{0x}+m_{\mathbb{S}\mid\mathbb{X}}(0\mid x)q_{1x}$
is the probability of a sale given $x$.

If $\sigma_{L}>0$ and $\sigma_{H}>0$, $\theta_{Q}(m)\equiv(\bar{s}_{L}(m),\bar{s}_{H}(m))$
is the unique parameter value that minimizes the wKLD function. If,
however, one of the prices is chosen with zero probability, there
are no restrictions on beliefs for the corresponding parameter, i.e.,
the set of minimizers is $\Theta_{Q}(m)=\{(\theta_{L},\theta_{H})\in\Theta:\theta_{H}=\bar{s}_{H}(m)\}$
if $\sigma_{L}=0$ and $\Theta_{Q}(m)=\{(\theta_{L},\theta_{H})\in\Theta:\theta_{L}=\bar{s}_{L}(m)\}$
if $\sigma_{H}=0$.

\emph{Stationary distribution}. Fix a strategy $\sigma_{H}$ and denote
a corresponding stationary distribution by $m(\cdot;\sigma_{H})\in\Delta(\mathbb{S}\times\mathbb{X})$.
Since the strategy does not depend on the state,  $m_{\mathbb{S}\mid\mathbb{X}}(\cdot\mid x;\sigma_{H})$
does not depend on $x$ and, therefore, coincides with the marginal
stationary distribution over $\mathbb{S}$, denoted by $m_{\mathbb{S}}(\cdot;\sigma_{H})\in\Delta(\mathbb{S})$.
This distribution is unique and given by the solution to  
\[
m_{\mathbb{S}}(1;\sigma_{H})=(1-m_{\mathbb{S}}(1;\sigma_{H}))((1-\sigma_{H})q_{0L}+\sigma_{H}q_{0H})+m_{\mathbb{S}}(1;\sigma_{H})((1-\sigma_{H})q_{1L}+\sigma_{H}q_{1H}).
\]

\emph{Equilibrium}. We restrict attention to equilibria that are robust
to experimentation (i.e., perfect equilibria; see Section \ref{sec:Equilibrium-refinements})
by focusing on the belief $\theta(\sigma_{H})=(\theta_{L}(\sigma_{H}),\theta_{H}(\sigma_{H}))\equiv\theta_{Q}(m(\cdot;\sigma_{H}))$
for a given strategy $\sigma_{H}\in[0,1]$.\footnote{Both $\sigma_{H}=0$ and $\sigma_{H}=1$ are Berk-Nash equilibria
supported by beliefs $\theta_{H}(0)=0$ and $\theta_{L}(1)=0$, respectively.
These outcomes, however, are not robust to experimentation, and are
eliminated by requiring $\theta_{H}(0)=\lim_{\sigma_{H}\rightarrow0}\bar{s}_{H}(m(\cdot;\sigma_{H}))=\bar{s}_{H}(m(\cdot;0))$,
and similarly for $\theta_{L}(1)$.} Next, let $\Delta(\theta(\sigma_{H}))$ be the \emph{perceived} expected
payoff difference for a given strategy $\sigma_{H}$. Note that $\sigma_{H}\mapsto\Delta(\theta(\sigma_{H}))$
is decreasing\footnote{The reason is that $\frac{d}{d\sigma_{H}}\Delta(\theta(\sigma_{H}))=\frac{d}{d\sigma_{H}}m_{\mathbb{S}}(1;\sigma_{H})\left(H(q_{1H}-q_{0H})+L(q_{1L}-q_{0L})\right)>0$,
since $\frac{d}{d\sigma_{H}}m_{\mathbb{S}}(1;\sigma_{H})<0$ and $q_{1x}>q_{0x}$
for all $x\in\{L,H\}$.}, which means that a higher probability of choosing price $H$ leads
to more pessimistic beliefs about the benefit of choosing $H$ vs.
$L$. Therefore, there exists a unique (perfect) equilibrium strategy.
Figure \ref{fig:monopolist} depicts an example where the equilibrium
is in mixed strategies.\footnote{See \cite{esponda2016berk} for the importance of mixed strategies
in misspecified settings.} Since $\Delta(\theta(0))>0$, an agent who always chooses a low price
must believe in equilibrium that setting a high price would instead
be optimal. Similarly, $\Delta(\theta(1))<0$ implies that an agent
who always chooses a high price must believe in equilibrium that settings
a low price would instead be optimal. Therefore, in equilibrium, the
agent chooses a strictly mixed strategy $\sigma_{H}^{*}\in(0,1)$
such that $\Delta(\theta(\sigma_{H}^{*}))=0$.\footnote{\label{fn:Def_D}More generally, the unique equilibrium is $\sigma_{H}=0$
if $\Delta(\theta(0))<0$ (i.e., $\frac{H}{L}\leq D_{1}\equiv\frac{q_{0L}}{(1-q_{1L})q_{0H}+q_{1H}q_{0L}}$),
$\sigma_{H}=1$ if $\Delta(\theta(1))>0$ (i.e., $\frac{H}{L}\geq D_{2}\equiv(1-q_{1H})\frac{q_{0L}}{q_{0H}}+q_{1L}$),
and $\sigma_{H}^{*}\in(0,1)$ the solution to $\Delta(\theta(\sigma_{H}^{*}))=0$
if $D_{1}<\frac{H}{L}<D_{2}$, where $\frac{q_{1L}}{q_{1H}}<D_{1}<D_{2}<\frac{q_{0L}}{q_{1H}}$.}

\begin{figure}
\centering{}\includegraphics[scale=0.6]{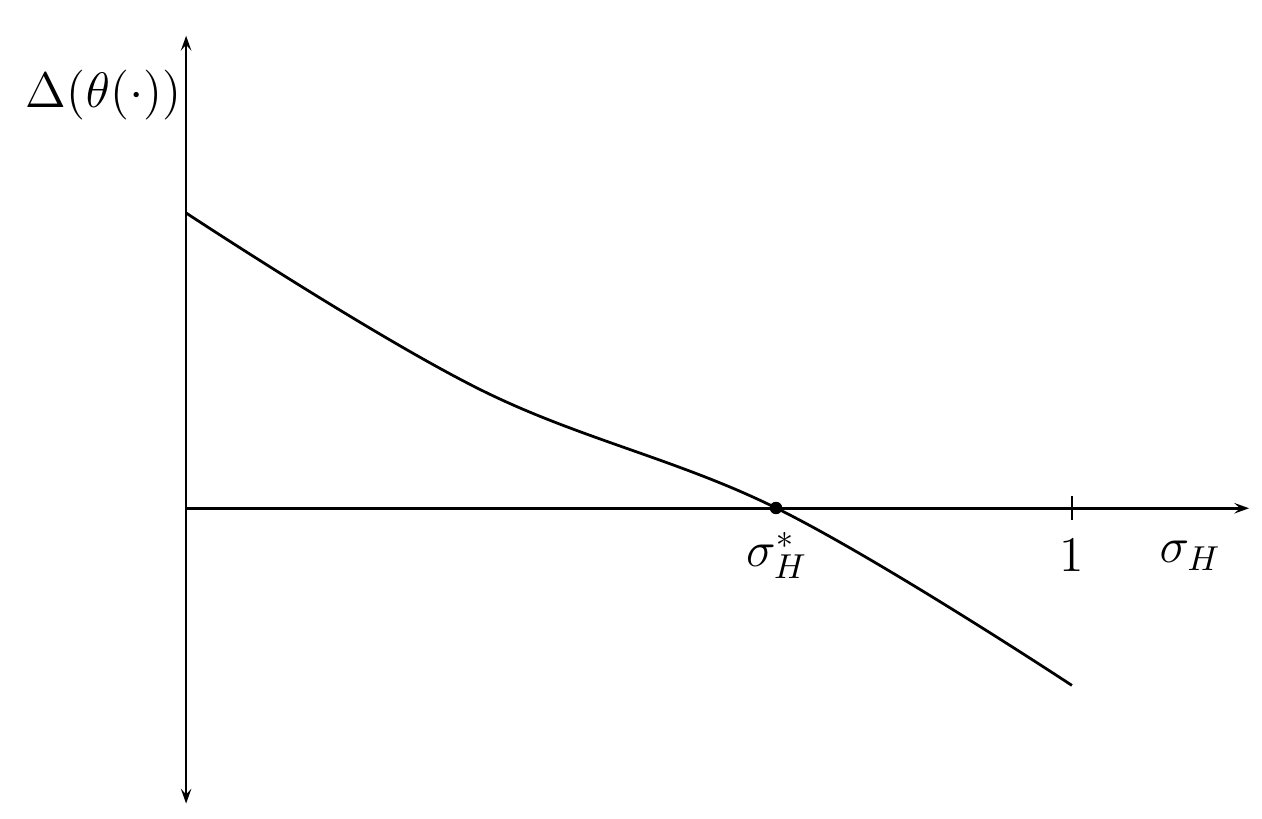}\protect\caption{Equilibrium of the monopoly example\label{fig:monopolist}}
\end{figure}

The misspecified monopolist may end up choosing higher prices than
optimal, since she fails to realize that high prices today cost her
in the future. But, a bit more surprisingly, she also may end up choosing
lower prices for some primitives.\footnote{This happens if $C_{\delta}<H/L<D_{1}$; see footnotes \ref{fn:C_delta}
and \ref{fn:Def_D}. } The reason is that her failure to realize that $H$ does relatively
better in state $s=1$ makes $H$ unattractive to her.

\subsection{Search with uncertainty about future job offers}

Search-theoretic models have been central to understanding labor markets
since \cite{mccall1970economics}. Most of the literature assumes
that the worker knows all the primitives. Exceptions include \cite{rothschild1974searching}
and \cite{burdett1988declining}, wherein the worker does not know
the wage distribution but has a correctly-specified model. In contrast,
we study a worker or entrepreneur who \emph{knows} the distribution
of wages or returns for new projects but does not know the probability
that she would be able to find a new job or fund a new project. The
worker or entrepreneur, however, does not realize that she is fired
or her project fails with higher probability in times in which it
is actually harder to find a new job or fund a new project. We show
that the worker or entrepreneur becomes pessimistic about the chances
of finding new prospects and sub-optimally accepts prospects with
low returns in equilibrium.

\textbf{MDP}. At the beginning of each period $t$, a worker (or entrepreneur)
faces a wage offer (or a project with returns) $w_{t}\in\mathbb{S}=[0,1]$
and decides whether to reject or accept it, $x_{t}\in\mathbb{X}=\{0,1\}$.\footnote{The set of feasible actions is independent of the state, i.e., $\Gamma(w)=\mathbb{X}$
for all $w\in\mathbb{S}$.} Her payoff in period $t$ is $\pi(w_{t},x_{t})=w_{t}x_{t}$; i.e,
she earns $w_{t}$ if she accepts and zero otherwise. After making
her decision, an economic fundamental $z_{t}\in\mathbb{Z}$ is drawn
from an i.i.d. distribution $G$.\footnote{To simplify the notation, we assume the fundamental is unobserved,
although the results are identical if it is observed, since it is
i.i.d. and it is realized after the worker makes her decision.} If the worker is employed, she is fired (or the project fails) with
probability $\gamma(z_{t})$. If the worker is unemployed (either
because she was employed and then fired or because she did not accept
employment at the beginning of the period), then with probability
$\lambda(z_{t})$ she draws a new wage $w_{t+1}\in[0,1]$ according
to some absolutely continuous distribution $F$ with density $f$;
wages are independent and identically distributed across time. With
probability $1-\lambda(z_{t})$, the unemployed worker receives no
wage offer, and we denote the corresponding state by $w_{t+1}=0$
without loss of generality. The worker will have to decide whether
to accept or reject $w_{t+1}$ at the beginning of next period. If
the worker accepted employment at wage $w_{t}$ at the beginning of
time $t$ and was not fired, then she starts next period with wage
offer $w_{t+1}=w_{t}$ and will again have to decide whether to quit
or remain in her job at that offer.\footnote{Formally, $Q(w'\mid w,x)$ is as follows: If $x=1$, then $w'=w$
with probability $1-\overline{\gamma}$, $w'$ is a draw from $F$
with probability $E[\gamma(Z)\lambda(Z)]$, and $w'=0$ with probability
$E[\gamma(Z)(1-\lambda(Z))]$; If $x=0$, then $w'$ is a draw from
$F$ with probability $\overline{\lambda}$ and $w'=0$ with probability
$1-\overline{\lambda}$.} The agent wants to maximize discounted expected utility with discount
factor $\delta\in[0,1)$. Suppose that $\overline{\gamma}\equiv E[\gamma(Z)]>0$
and $\overline{\lambda}\equiv E[\lambda(Z)]>0$.

We assume that $Cov(\gamma(Z),\lambda(Z))<0$; for example, the worker
is more likely to get fired and less likely to receive an offer when
economic fundamentals are strong, and the opposite holds when fundamentals
are weak.

\textbf{SMDP}. The worker knows all the primitives except $\lambda(\cdot)$,
which determines the probability of receiving an offer. The worker
has a misspecified model of the world and believes $\lambda(\cdot)$
does not depend on the economic fundamental, i.e., $\lambda(z)=\theta$
for all $z\in\mathbb{Z}$ , where $\theta\in[0,1]$ is the unknown
parameter.\footnote{The results are identical if the agent is also uncertain of $\gamma(\cdot)$;
given the current misspecification, the agent only cares about the
expectation of $\gamma$ and will have correct beliefs about it.} The transition probability function $Q_{\theta}(w'\mid w,x)$ is
as follows: If $x=1$, then $w'=w$ with probability $1-\theta$,
$w'$ is a draw from $F$ with probability $\theta\overline{\gamma}$,
and $w'=0$ with probability $(1-\theta)\overline{\gamma}$; If $x=0$,
then $w'$ is a draw from $F$ with probability $\theta$ and $w'=0$
with probability $1-\theta$.

\textbf{Equilibrium}. \emph{Optimality}. Suppose that the worker believes
that the true parameter is $\theta$ with probability 1. The value
of receiving wage offer $w\in\mathbb{S}$ is 
\begin{align*}
V(w) & =\max\left\{ w+\delta\left((1-\overline{\gamma})V(w)+(1-\theta)\overline{\gamma}V(0)+\theta\overline{\gamma}E[V(W')]\right),\right.\\
 & \left.\,\,\,\,\,\,\,\,\,\,\,\,\,\,\,\,\,\,\,\,\,\,0+\delta\left(\theta E[V(W')]+(1-\theta)V(0)\right)\right\} .
\end{align*}
 By standard arguments, her optimal strategy is a stationary reservation
wage strategy $w(\theta)$ that solves the following equation: 
\begin{equation}
w(\theta)(1-\delta+\delta\overline{\gamma})=\delta\theta(1-\overline{\gamma})\int_{w>w(\theta)}\left(w-w(\theta)\right)F(dw).\label{eq:wage(theta)}
\end{equation}
The worker accepts wages above the reservation wage and rejects wages
below it. Also, $\theta\mapsto w(\theta)$ is increasing: The higher
is the probability of receiving a wage offer, then the more she is
willing to wait for a better offer in the future. Figure \ref{fig:search}
depicts an example.

\emph{Beliefs}. For any $m\in\Delta(\mathbb{S}\times\mathbb{X})$,
the wKLD simplifies to 
\begin{align*}
K_{Q}(m,\theta) & =\int_{\mathbb{S}\times\mathbb{X}}E_{Q(\cdot\mid\tilde{w},x)}\Bigl[\ln\frac{Q(W'\mid\tilde{w},x)}{Q_{\theta}(W'\mid\tilde{w},x)}\Bigr]m(d\tilde{w},dx)\\
 & =\Bigl\{ E[\gamma\lambda]\ln\frac{E[\gamma\lambda]}{\overline{\gamma}\theta}+E[\gamma(1-\lambda)]\ln\frac{E[\gamma(1-\lambda)]}{\overline{\gamma}(1-\theta)}\Bigr\} m_{\mathbb{X}}(1)\\
 & +\Bigl\{\overline{\lambda}\ln\frac{\overline{\lambda}}{\theta}+(1-\overline{\lambda})\ln\frac{1-\overline{\lambda}}{1-\theta}\Bigr\} m_{\mathbb{X}}(0),
\end{align*}
where the density of $W'$ cancels out because the workers knows it
and where $m_{\mathbb{X}}$ is the marginal distribution over $\mathbb{X}$.
In the Online Appendix, we show that the unique parameter that minimizes
$K_{Q}(m,\cdot)$ is 
\begin{equation}
\theta_{Q}(m)\equiv\frac{m_{\mathbb{X}}(0)}{m_{\mathbb{X}}(0)+m_{\mathbb{X}}(1)\overline{\gamma}}\overline{\lambda}+\Bigl(1-\frac{m_{\mathbb{X}}(0)}{m_{\mathbb{X}}(0)+m_{\mathbb{X}}(1)\overline{\gamma}}\Bigr)\Bigl(\overline{\lambda}+\frac{Cov(\gamma,\lambda)}{\overline{\gamma}}\Bigr).\label{eq:eqmtheta*_mis}
\end{equation}
To see the intuition behind equation (\ref{eq:eqmtheta*_mis}), note
that the agent only observes the realization of $\lambda$, i.e.,
whether she receives a wage offer, when she is unemployed. Unemployment
can be voluntary or involuntary. In the first case, the agent rejects
the offer and, since this decision happens before the fundamental
is realized, it is independent of getting or not an offer. Thus, with
conditional on unemployment being voluntary, the agent will observe
an unbiased average probability of getting an offer, $\overline{\lambda}$
(see the first term in the RHS of (\ref{eq:eqmtheta*_mis})). In the
second case, the agent accepts the offer but is then fired. Since
$Cov(\gamma,\lambda)<0$, she is less likely to get an offer in periods
in which she is fired and, because she does not account for this correlation,
she will have a more pessimistic view about the probability of receiving
a wage offer relative to the average probability $\overline{\lambda}$
(the second term in the RHS of (\ref{eq:eqmtheta*_mis}) captures
this bias).

\emph{Stationary distribution}. Fix a reservation wage strategy $w$
and denote the marginal over $\mathbb{X}$ of the corresponding stationary
distribution by $m_{\mathbb{X}}(\cdot;w)\in\Delta(\mathbb{X})$. In
the Online Appendix, we characterize $m_{\mathbb{X}}(\cdot;w)$ and
show that $w\mapsto m_{\mathbb{X}}(0;w)$ is increasing. Intuitively,
the more selective the worker, the higher the chance of being unemployed.

\emph{Equilibrium}. Let $\theta(\omega)\equiv\theta_{Q}(m(\cdot;w))$
denote the equilibrium belief for an agent following reservation wage
strategy $w$. The weight on $\overline{\lambda}$ in equation (\ref{eq:eqmtheta*_mis})
represents the probability of voluntary unemployment conditional on
unemployment. This weight is increasing in $\omega$ because $w\mapsto m_{\mathbb{X}}(0;w)$
is increasing. Therefore, $w\mapsto\theta(w)$ is increasing. In the
extreme case in which $w=1$, the worker rejects all offers, unemployment
is always voluntary, and the bias disappears, $\theta(1)=\overline{\lambda}$.
An example of the schedule $\theta(\cdot)$ is depicted in Figure
\ref{fig:search}. The set of Berk-Nash equilibria is given by the
intersection of $w(\cdot)$ and $\theta(\cdot)$. In the example depicted
in Figure \ref{fig:search}, there is a unique equilibrium strategy
$w^{M}=w(\theta^{M})$, where $\theta^{M}<\overline{\lambda}$.

\begin{figure}
\centering{}\includegraphics[scale=0.6]{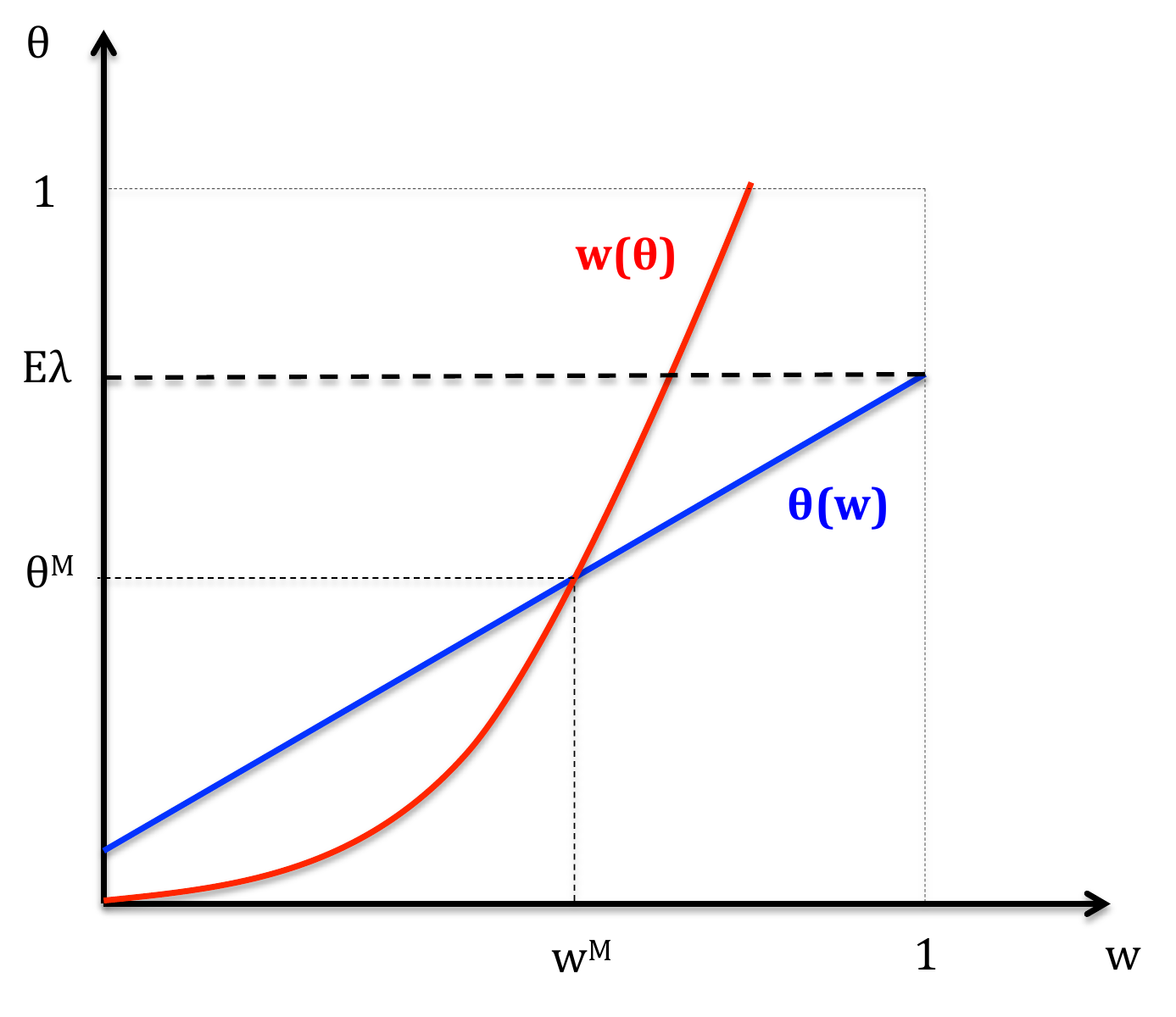}\protect\caption{Equilibrium of the search model\label{fig:search}}
\end{figure}

We conclude by comparing Berk-Nash equilibria to the optimal strategy
of a worker who knows the primitives, $w^{*}$. By standard arguments,
$w^{*}$ is the unique solution to 
\begin{equation}
w^{\ast}(1-\delta+\delta\overline{\gamma})=\delta(\overline{\lambda}-E[\gamma\lambda])\int_{w>w^{\ast}}\left(w-w^{\ast}\right)F(dw).\label{eq:w*-well-1}
\end{equation}
The only difference between equations (\ref{eq:wage(theta)}) and
(\ref{eq:w*-well-1}) appears in the term multiplying the RHS, which
captures the cost of accepting a wage offer. In the misspecified case,
this term is $\delta\theta(1-\overline{\gamma})$; in the correct
case, it is $\delta(\overline{\lambda}-E[\gamma\lambda])=\delta\overline{\lambda}(1-\overline{\gamma})-\delta Cov(\gamma,\lambda)$.
The misspecification affects the optimal threshold in two ways. First,
the misspecified agent estimates the mean of $\lambda$ incorrectly,
i.e., $\theta<\overline{\lambda}$; therefore, she (incorrectly) believes
that, in expectation, offers arrive with lower probability. Second,
she does not realize that, because $Cov(\gamma,\lambda)<0$, she is
less likely to receive an offer when fired. Both effects go in the
same direction and make the option to reject and wait for the possibility
of drawing a new wage offer next period less attractive for the misspecified
worker. Formally, $\theta\delta(1-\overline{\gamma})<\delta\overline{\lambda}(1-\overline{\gamma})-\delta Cov(\gamma,\lambda)$
and so $w^{M}<w^{*}$.

\subsection{Stochastic growth with correlated shocks}

Stochastic growth models have been central to studying optimal intertemporal
allocation of capital and consumption since the work of \cite{brock1972optimal}.
\cite{freixas1981optimal} and \cite{koulovatianos2009optimal}
assume that agents learn the distribution over productivity shocks
with correctly specified models. We follow \cite{hall1997macroeconomic}
and subsequent literature in incorporating shocks to both preferences
and productivity, but assume that these shocks are (positively) correlated.
We show that agents who fail to account for the correlation of shocks
underinvest in equilibrium.

\textbf{MDP}. In each period $t$, an agent observes $s_{t}=(y_{t},z_{t})\in\mathbb{S}=\mathbb{R}_{+}\times\{L,H\}$,
where $y_{t}$ is income from the previous period and $z_{t}$ is
a current utility shock, and chooses how much income to save, $x_{t}\in\Gamma(y_{t},z_{t})=[0,y_{t}]\subseteq\mathbb{X}=\mathbb{R}_{+}$,
consuming the rest. Current period utility is $\pi(y_{t},z_{t},x_{t})=z_{t}\ln(y_{t}-x_{t})$.
Income next period, $y_{t+1}$, is given by 
\begin{equation}
\ln y_{t+1}=\alpha^{*}+\beta^{*}\ln x_{t}+\varepsilon_{t},\label{eq:production}
\end{equation}
where $\varepsilon_{t}=\gamma^{*}z_{t}+\xi_{t}$ is an unobserved
productivity shock, $\xi_{t}\sim N(0,1)$, and $0<\delta\beta^{*}<1$,
where $\delta\in[0,1)$ is the discount factor. We assume that $\gamma^{*}>0$,
so that the utility and productivity shocks are positively correlated.
Let $0<L<H$ and let $q\in(0,1)$ be the probability that the shock
is $H$.\footnote{Formally, $Q(y',z'\mid y,z,x)$ is such that $y'$ and $z'$ are independent,
$y'$ has a log-normal distribution with mean $\alpha^{*}+\beta^{*}\ln x+\gamma^{*}z$
and unit variance, and $z'=H$ with probability $q$.}

\textbf{SMDP}. The agent believes that 
\begin{equation}
\ln y_{t+1}=\alpha+\beta\ln x_{t}+\varepsilon_{t},\label{eq:production_subjective}
\end{equation}
where $\varepsilon_{t}\sim N(0,1)$ and is \emph{independent} of the
utility shock. For simplicity, we assume that the agent knows the
distribution of the utility shock, and is uncertain about $\theta=(\alpha,\beta)\in\Theta=\mathbb{R}^{2}$.
The subjective transition probability function $Q_{\theta}(y',z'\mid y,z,x)$
is such that $y'$ and $z'$ are independent, $y'$ has a log-normal
distribution with mean $\alpha+\beta\ln x$ and unit variance, and
and $z'=H$ with probability $q$. The agent has a misspecified model
because she believes that the productivity and utility shocks are
independent when in fact $\gamma^{*}\neq0$.

\textbf{Equilibrium}. \emph{Optimality}. The Bellman equation for
the agent is 
\[
V(y,z)=\max_{0\leq x\leq y}z\ln(y-x)+\delta E\left[V(Y',Z')\mid x\right]
\]
and it is straightforward to verify that the optimal strategy is to
invest a fraction of income that depends on the utility shock and
the unknown parameter $\beta$, i.e., $x=A_{z}(\beta)\cdot y$, where
$A_{L}(\beta)=\frac{\delta\beta((1-q)L+qH)}{(1-\delta\beta(1-q))H+\delta\beta(1-q)}$
and $A_{H}(\beta)=\frac{\delta\beta((1-q)L+qH)}{\delta\beta qH+(1-\delta\beta q)L}<A_{L}(\beta)$.
For the agent who knows the primitives, the optimal strategy is to
invest fractions $A_{L}(\beta^{*})$ and $A_{H}(\beta^{*})$ in the
low and high state, respectively. Since $\beta\mapsto A_{z}(\beta)$
is increasing, the equilibrium strategy of a misspecified agent can
be compared to the optimal strategy by comparing the equilibrium belief
about $\beta$ with the true $\beta^{*}$.

\emph{Beliefs} \emph{and stationary distribution}. Let $A=(A_{L},A_{H})$,
with $A_{H}<A_{L}$, represent a strategy, where $A_{z}$ is the proportion
of income invested given utility shock $z$. Because the agent believes
that $\varepsilon_{t}$ is independent of the utility shock and normally
distributed, minimizing the wKLD function is equivalent to performing
an OLS regression of equation (\ref{eq:production_subjective}). Thus,
for a strategy represented by $A=(A_{L},A_{H})$, the parameter value
$\hat{\beta}(A)$ that minimizes wKLD is 
\begin{align*}
\hat{\beta}(A) & =\frac{Cov(\ln Y',\ln X)}{Var(\ln X)}=\frac{Cov(\ln Y',\ln A_{Z}Y)}{Var(\ln A_{Z}Y)}\\
 & =\beta^{*}+\gamma^{*}\frac{Cov(Z,\ln A_{Z})}{Var(\ln A_{Z})+Var(Y)}.
\end{align*}
where $Cov$ and $Var$ are taken with respect to the (true) stationary
distribution of $(Y,Z)$. Since $A_{H}<A_{L}$, then $Cov(Z,\ln A_{Z})<0$.
Therefore, the assumption that $\gamma^{*}>0$ implies that the bias
$\hat{\beta}(A)-\beta^{*}$ is negative and its magnitude depends
on the strategy $A$. Intuitively, the agent invests a larger fraction
of income when $z$ is low, which happens to be during times when
$\varepsilon$ is also low.

\emph{Equilibrium}. We establish that there exists at least one equilibrium
with positive investment by showing that there is at least one fixed
point of the function $\hat{\beta}(A_{L}(\beta),A_{H}(\beta))$.\footnote{Our existence theorem is not directly applicable because we have assumed,
for convenience, nonfinite state and action spaces.} The function is continuous in $\beta$ and satisfies $\hat{\beta}(A_{L}(0),A_{H}(0))=\hat{\beta}(A_{L}(1/\delta),A_{H}(1/\delta))=\beta^{*}$
and $\hat{\beta}(A_{L}(\beta),A_{H}(\beta))<\beta^{*}$ for all $\beta\in(0,1/\delta)$.
Then, since $\delta\beta^{*}<1$, there is at least one fixed point
$\beta^{M}$, and any fixed point satisfies $\beta^{M}\in(0,\beta^{*})$.
Thus, the misspecified agent underinvests in equilibrium compared
to the optimal strategy.\footnote{It is also an equilibrium not to invest, $A=(0,0)$, supported by
the belief $\beta^{*}=0$, which cannot be disconfirmed since investment
does not take place. But this equilibrium is not robust to experimentation
(i.e., it is not perfect; see Section \ref{sec:Equilibrium-refinements}).} The conclusion is reversed if $\gamma^{*}<0$, illustrating how the
framework provides predictions about beliefs and behavior that depend
on the primitives (as opposed to simply postulating that the agent
is over or under-confident about productivity).

\section{\label{sec:Equilibrium-foundation}Equilibrium foundation}

In this section, we provide a learning foundation for the notion of
Berk-Nash equilibrium of SMDPs. We fix an SMDP and assume that the
agent is Bayesian and starts with a prior $\mu_{0}\in\Delta(\Theta)$
over her set of models of the world. She observes past actions and
states and uses this information to update her beliefs about $\Theta$
in every period.\medskip{}

\begin{defn}
\label{Def:Bayesian}For any $(s,x,s')\in Gr(\Gamma)\times\mathbb{S}$,
let $B(s,x,s',\cdot):D_{s,x,s'}\rightarrow\Delta(\Theta)$ denote
the Bayesian operator: For all $A\subseteq\Theta$ Borel 
\begin{equation}
B(s,x,s',\mu)(A)=\frac{\int_{A}Q_{\theta}(s'\mid s,x)\mu(d\theta)}{\int_{\Theta}Q_{\theta}(s'\mid s,x)\mu(d\theta)}.\label{eq:bayesian_operator-1-1}
\end{equation}
 for any $\mu\in D_{s,x,s'}$, where $D_{s,x,s'}=\{p\in\Delta(\Theta)\colon\int_{\Theta}Q_{\theta}(s'\mid s,x)p(d\theta)>0\}$.%

\end{defn}
\medskip{}

\begin{defn}
A \textbf{Bayesian Subjective Markov Decision Process} (Bayesian-SMDP)
is an SMDP($Q,\mathcal{Q}_{\Theta}$) together with a prior $\mu_{0}\in\Delta(\Theta)$
and the Bayesian operator $B$ (see Definition \ref{Def:Bayesian}).
It is said to be \textbf{regular} if the corresponding SMDP is regular.
\end{defn}
\medskip{}

By the Principle of Optimality, the agent's problem in a Bayesian-SMDP
can be cast recursively as 
\begin{equation}
W(s,\mu)=\max_{x\in\Gamma(s)}\int_{\mathbb{S}}\left\{ \pi(s,x,s')+\delta W(s',\mu')\right\} \bar{Q}_{\mu}(ds'|s,x),\label{eq:bellman-perturbedBMDP}
\end{equation}
where $\bar{Q}_{\mu}=\int_{\Theta}Q_{\theta}\mu(d\theta)$, $\mu'=B(s,x,s',\mu)$
is next period's belief, updated using Bayes' rule, and $W:\mathbb{S}\times\Delta(\Theta)\rightarrow\mathbb{R}$
is the (unique) solution to the Bellman equation (\ref{eq:bellman-perturbedBMDP}).
Compared to the case where the agent knows the transition probability
function, the agent's belief about $\Theta$ is now part of the state
space.

\medskip{}

\begin{defn}
A \textbf{policy function} is a function $f:\Delta(\Theta)\rightarrow\Sigma$
mapping beliefs into strategies (recall that a strategy is a mapping
$\sigma:\mathbb{S}\rightarrow\Delta(\mathbb{X})$). For any belief
$\mu\in\Delta(\Theta)$, state $s\in\mathbb{S}$, and action $x\in\mathbb{X}$,
let $f(x\mid s,\mu)$ denote the probability that the agent chooses
$x$ when selecting policy function $f$. A policy function $f$ is
\textbf{optimal} for the Bayesian-SMDP if, for all $s\in\mathbb{S}$,
$\mu\in\Delta(\Theta)$, and $x\in\mathbb{X}$ such that $f(x\mid s,\mu)>0$,
\[
x\in\arg\max_{\hat{x}\in\Gamma(s)}\int_{\mathbb{S}}\left\{ \pi(s,\hat{x},s')+\delta W(s',\mu')\right\} \bar{Q}_{\mu}(ds'|s,\hat{x}).
\]

\end{defn}
\medskip{}

For each $\mu\in\Delta(\Theta)$, let $\bar{\Sigma}(\mu)\subseteq\Sigma$
denote the set of all strategies that are induced by a policy that
is \emph{optimal}, i.e., 
\[
\bar{\Sigma}(\mu)=\bigl\{\sigma\in\Sigma:\exists\,\,\mbox{optimal }f\,\,\mbox{such that \ensuremath{\sigma(\cdot\mid s)=f(\cdot\mid s,\mu)}for all \ensuremath{s\in\mathbb{S}}}\bigr\}.
\]

\medskip{}

\begin{lem}
\label{Lemma:Sigma(Q)-1}(i) There is a unique solution $W$ to the
Bellman equation in (\ref{eq:bellman-perturbedBMDP}), and it is continuous
in $\mu$ for all $s\in\mathbb{S}$; (ii) The correspondence of optimal
strategies $\mu\mapsto\bar{\Sigma}(\mu)$ is non-empty, compact-valued,
convex-valued, and upper hemicontinuous.\end{lem}
\begin{proof}
The proof is standard and relegated to the Online Appendix.
\end{proof}
\medskip{}

Let $h^{\infty}=(s_{0},x_{0},...,s_{t},x_{t},...)$ represent the
infinite history or outcome path of the dynamic optimization problem
and let $\mathbb{H}^{\infty}\equiv(Gr(\Gamma))^{\infty}$ represent
the space of infinite histories. For every $t$, let $\mu_{t}:\mathbb{H}^{\infty}\rightarrow\Delta(\Theta)$
denote the agent's Bayesian beliefs, defined recursively by $\mu_{t}=B(s_{t-1},x_{t-1},s_{t},\mu_{t-1})$
whenever $\mu_{t-1}\in D_{s_{t-1},x_{t-1},s_{t}}$ (see Definition
\ref{Def:Bayesian}), and arbitrary otherwise. We assume that the
agent follows some policy function $f$. %
In each period $t$, there is a state $s_{t}$ and a belief $\mu_{t}$,
and the agent chooses a (possibly mixed) action $f(\cdot\mid s_{t},\mu_{t})\in\Delta(\mathbb{X})$.
After an action $x_{t}$ is realized, the state $s_{t+1}$ is drawn
from the true transition probability. The agent observes the realized
action and the new state and updates her beliefs to $\mu_{t+1}$ using
Bayes' rule. The primitives of the Bayesian-SMDP (including the initial
distribution over states, $q_{0}$, and the prior, $\mu_{0}\in\Delta(\Theta)$)
and a policy function $f$ induce a probability distribution over
$\mathbb{H}^{\infty}$ that is defined in a standard way; let $\boldsymbol{P}^{f}$
denote this probability distribution over $\mathbb{H}^{\infty}$.

We now define strategies and outcomes as random variables. For a fixed
policy function $f$ and for every $t$, let $\sigma_{t}:\mathbb{H}^{\infty}\rightarrow\Sigma$
denote the strategy of the agent, defined by setting 
\[
\sigma_{t}(h^{\infty})=f(\cdot\mid\cdot,\mu_{t}(h^{\infty}))\in\Sigma.
\]
Finally, for every $t$, let $m_{t}:\mathbb{H}^{\infty}\rightarrow\Delta(Gr(\Gamma))$
be such that, for all $t$, $h^{\infty}$, and $(s,x)\in Gr(\Gamma)$,
\[
m_{t}(s,x\mid h^{\infty})=\frac{1}{t}\sum_{\tau=0}^{t}\mathbf{1}_{(s,x)}(s_{\tau},x_{\tau})
\]
is the frequency of times that the outcome $(s,x)$ occurs up to time
$t$. 

One reasonable criteria to claim that the agent has reached a steady-state
is that her strategy and the time average of outcomes converge.\medskip{}

\begin{defn}
\label{Def:stability}A strategy and probability distribution $(\sigma,m)\in\Sigma\times\Delta(Gr(\Gamma))$
is \textbf{stable} for a Bayesian-SMDP with prior $\mu_{0}$ and
policy function $f$ if there is a set $\mathcal{H}\subseteq\mathbb{H}$
with $\mathbf{P}^{f}\left(\mathcal{H}\right)>0$ such that, for all
$h^{\infty}\in\mathcal{H}$, as $t\rightarrow\infty$,
\begin{equation}
\sigma_{t}(h^{\infty})\rightarrow\sigma\,\,\,\,\,\mbox{and}\,\,\,\,\,m_{t}(h^{\infty})\rightarrow m.\label{eq:sigma_converges}
\end{equation}
If, in addition, there exists a belief $\mu^{*}$ and a subsequence
$(\mu_{t(j)})_{j}$ such that, 
\begin{equation}
\mu_{t(j)}(h^{\infty})\overset{w}{\rightarrow}\mu^{*}\label{eq:beliefs_converge}
\end{equation}
and, for all $(s,x)\in Gr(\Gamma)$, $\mu^{*}=B(s,x,s',\mu^{*})$
for all $s'\in\mathbb{S}$ such that $\bar{Q}_{\mu^{*}}(s'\mid s,x)>0$,
then $(\sigma,m)$ is called \textbf{stable with exhaustive learning}.

\end{defn}
\medskip{}
Condition (\ref{eq:sigma_converges}) requires that strategies and
the time frequency of outcomes stabilize. %
{} By compactness, there exists a subsequence of beliefs that converges.
The additional requirement of exhaustive learning says that the limit
point of one of the subsequences, $\mu^{*}$, is perceived to be a
fixed point of the Bayesian operator, implying that no matter what
state and strategy the agent contemplates, she does not \emph{expect}
her belief to change. Thus, the agent believes that all learning possibilities
are exhausted under $\mu^{*}$. The condition, however, does not imply
that the agent has correct beliefs in steady state.

The next result establishes that, if the time average of outcomes
stabilize to $m$, then beliefs become increasingly concentrated on
$\Theta_{Q}(m)$. 

\medskip{}

\begin{lem}
\label{Lemma:Berk}Consider a regular Bayesian-SMDP with true transition
probability function $Q$, full-support prior $\mu_{0}\in\Delta(\Theta)$,
and policy function $f$. Suppose that $(m_{t})_{t}$ converges to
$m$ for all histories in a set $\mathcal{H}\subseteq\mathbb{H}$
such that $\mathbf{P}^{f}\left(\mathcal{H}\right)>0$. Then, for all
open sets $U\supseteq\Theta_{Q}(m)$, 
\[
\lim_{t\rightarrow\infty}\mu_{t}\left(U\right)=1
\]
$\mathbf{P}^{f}$-a.s. in $\mathcal{H}$.\end{lem}
\begin{proof}
See the Appendix.
\end{proof}
\medskip{}

The proof of Lemma \ref{Lemma:Berk} clarifies the origin of the wKLD
function in the definition of Berk-Nash equilibrium. The proof adapts
the proof of Lemma 2 by \cite{esponda2016berk} to dynamic environments.
Lemma \ref{Lemma:Berk} extends results from the statistics of misspecified
learning (\cite{berk1966limiting}, \cite{bunke1998asymptotic},
\cite{shalizi2009dynamics}) by considering a setting where agents
learn from data that is endogenously generated by their own actions
in a Markovian setting.

The following result provides a learning foundation for the notion
of Berk-Nash equilibrium of an SMDP.\medskip{}

\begin{thm}
\label{theo:Stability_implies_equilibrium}There exists $\bar{\delta}\in[0,1]$
such that: 

(i) for all $\delta\leq\bar{\delta}$, if $(\sigma,m)$ is stable
for a regular Bayesian-SMDP with full-support prior $\mu_{0}$ and
policy function $f$ that is optimal, then $(\sigma,m)$ is a Berk-Nash
equilibrium of the SMDP.

(ii) for all $\delta>\bar{\delta}$, if $(\sigma,m)$ is stable with
exhaustive learning for a regular Bayesian-SMDP with full-support
prior $\mu_{0}$ and policy function $f$ that is optimal, then $(\sigma,m)$
is a Berk-Nash equilibrium of the SMDP.\end{thm}
\begin{proof}
See the Appendix.
\end{proof}
\medskip{}

Theorem \ref{theo:Stability_implies_equilibrium} provides a learning
justification for Berk-Nash equilibrium. The main idea behind the
proof is as follows. We can always find a subsequence of posteriors
that converges to some $\mu^{*}$ and, by Lemma \ref{Lemma:Berk}
and the fact that behavior converges to $\sigma$, it follows that
$\sigma$ must solve the dynamic optimization problem for beliefs
converging to $\mu^{*}\in\Theta_{Q}(m)$. In addition, by convergence
of $\sigma_{t}$ to $\sigma$ and continuity of the transition kernel
$\sigma\mapsto M_{\sigma,Q}$, an application of the martingale convergence
theorem implies that $m_{t}$ is asymptotically equal to $M_{\sigma,Q}[m_{t}]$.
This fact, linearity of the operator $M_{\sigma,Q}[\cdot]$, and convergence
of $m_{t}$ to $m$ then imply that $m$ is an invariant distribution
given $\sigma$.

The proof concludes by showing that $\sigma$ not only solves the
optimization problem for beliefs converging to $\mu^{*}$ but also
solves the MDP, where the belief is forever fixed at $\mu^{*}$. This
is true, of course, if the agent is sufficiently impatient, which
explains why part (i) of Theorem \ref{theo:Stability_implies_equilibrium}
holds. For sufficiently patient agents, the result relies on the assumption
that the steady state satisfies exhaustive learning. We now illustrate
and discuss the role of this assumption.

\medskip{}

$\textsc{example}$. At the initial period, a risk-neutral agent has
four investment choices: A, B, S, and O. Action A pays $1-\theta^{*}$,
action B pays $\theta^{*}$, and action S pays a safe payoff of 2/3
in the initial period, where $\theta^{*}\in\{0,1\}$. For any of these
three choices, the decision problem ends there and the agent makes
a payoff of zero in all future periods. Action O gives the agent a
payoff of $-1/3$ in the initial period and the option to make an
investment next period, where there are two possible states, $s_{A}$
and $s_{B}$. State $s_{A}$ is realized if $\theta^{*}=1$ and state
$s_{B}$ is realized if $\theta^{*}=0$. In each of these states,
the agent can choose to make a risky investment or a safe investment.
The safe investment gives a payoff of 2/3 in both states, and a subsequent
payoff of zero in all future periods. The risky investment gives the
agent a payoff that is thrice the payoff she would have gotten from
choice A, that is, $3(1-\theta^{*})$, if the state is $s_{A}$, and
it gives the agent thrice the payoff she would have gotten from choice
B, that is, $3\theta^{*}$, if the state is $s_{B}$; the payoff is
zero is all future periods.

Suppose that the agent knows all the primitives except the value of
$\theta^{*}$. Let $\Theta=\{0,1\}$; in particular, the SMDP is correctly
specified. We now show that, in any Berk-Nash equilibrium, a sufficiently
patient agent never chooses the safe action S: Let $\mu\in[0,1]$
denote the agent's equilibrium belief about the probability that $\theta^{*}=1$.
For action S to be preferred to A and B, it must be the case that
$\mu\in[1/3,2/3]$. But, for a fixed $\mu$, the perceived benefit
from action O is 
\begin{align*}
-\frac{1}{3}+\delta\left(\mu V_{\bar{Q}_{\mu}}(s_{A})+(1-\mu)V_{\bar{Q}_{\mu}}(s_{B})\right) & =-\frac{1}{3}+\delta\left(\mu\max\{\frac{2}{3},3(1-\mu)\}+(1-\mu)\max\{\frac{2}{3},3\mu\}\right)\\
 & \geq-\frac{1}{3}+\delta6\mu(1-\mu),
\end{align*}
which is strictly higher than 2/3, the payoff from action S, for all
$\mu\in[1/3,2/3]$ provided that $\delta>\bar{\delta}=3/4$. Thus,
for a sufficiently patient agent, there is no belief that makes action
S optimal and, therefore, S is not chosen in any Berk-Nash equilibrium.

Now consider a Bayesian agent who starts with a prior $\mu_{0}=\Pr(\theta=1)\in(0,1)$
and updates her belief. The value of action O is 
\[
-\frac{1}{3}+\delta\left(\mu_{0}W(s_{A},1)+(1-\mu_{0})W(s_{B},0)\right)=-\frac{1}{3}+\delta\frac{2}{3}<\frac{2}{3}
\]
because $W(s_{A},1)=W(s_{B},0)=2/3$. In other words, the agent realizes
that if the state $s_{A}$ is realized, then she will update her belief
to $\mu_{1}=1$, which implies that the safe investment is optimal
in state $s_{A}$; a similar argument holds for state $s_{B}$. She
then finds it optimal to choose action A if $\mu_{0}\leq1/3$, B if
$\mu_{0}\geq2/3$, and S if $\mu_{0}\in[1/3,2/3]$. In particular,
choosing S is a steady state outcome for some priors, although it
is not chosen in any Berk-Nash equilibrium if the agent is sufficiently
patient. The belief supporting S, however, does not satisfy exhaustive
learning, since the agent believes that any other action would completely
reveal all uncertainty. $\square$

\medskip{}

More generally, the failure of a steady state to be a Berk-Nash equilibrium
if the agent is sufficiently patient occurs because the value of experimentation
can be negative. To see this point, let the value of experimentation
for action $x$ at state $s$ when the agent's belief is $\mu$ be
\[
ValueExp(s,x;\mu)\equiv E_{\bar{Q}_{\mu}(\cdot\mid s,x)}\left[W(S',B(s,x,S',\mu))\right]-E_{\bar{Q}_{\mu}(\cdot\mid s,x)}\bigl[V_{\bar{Q}_{\mu}}(S')\bigr].
\]
This expression is the difference between the value when the agent
updates her prior $\mu$ and the value when the agent has a fixed
belief $\mu$. An agent who does not account for future changes in
beliefs may end up choosing an action with a negative value of experimentation
that is actually suboptimal when accounting for changes in beliefs.

In the previous example, the value of experimentation for action O
given $\mu$ is

\[
\left(\mu W(s_{A},1)+(1-\mu)W(s_{B},0)\right)-\left(\mu V_{\bar{Q}_{\mu}}(s_{A})+(1-\mu)V_{\bar{Q}_{\mu}}(s_{B})\right),
\]
which reduces to $2/3-6\mu(1-\mu)$ and is negative for the values
of $\mu$ that make S better than A and B. Thus, it is possible for
action O to be optimal if the agent does not account for changes in
beliefs, but suboptimal if she does.

We now discuss specifically how the property of exhaustive learning
is used in the proof of Theorem \ref{theo:Stability_implies_equilibrium}.
We call an action a steady-state action if it is in the support of
a stable strategy and we call it a non steady-state action otherwise.
A key step is to show that, if a steady-state action is better than
a non steady-state action when beliefs are updated, it will also be
better when beliefs are fixed. This is true provided that there is
zero value of experimenting in steady state, which is guaranteed by
exhaustive learning. If instead of exhaustive learning we were to
simply require weak identification, there would be no value of experimentation
for steady-state actions. The concern, illustrated by the previous
example, is that the value of experimentation can be negative for
a non steady-state action. Therefore, a non steady-state action could
be suboptimal in the problem where the belief is updated but optimal
in the problem where the belief is not updated (and so the negative
value of experimentation is not taken into account). As shown by \cite{esponda2016berk},
this concern does not arise in static settings, where the only state
variable is a belief. The reason is that the convexity of the value
function and the martingale property of Bayesian beliefs imply that
the value of experimentation is always nonnegative.

We conclude with additional remarks about Theorem \ref{theo:Stability_implies_equilibrium}.\medskip{}

\begin{rem}
\emph{Discount factor}: In the proof of Theorem \ref{theo:Stability_implies_equilibrium},
we provide an exact value for $\bar{\delta}$ as a function of primitives.
This bound, however, may not be sharp. As illustrated by the above
example, to compute a sharp bound we would have to solve the dynamic
optimization problem with learning, which is precisely what we are
trying to avoid by focusing on Berk-Nash equilibrium.

\emph{Convergence}: Theorem \ref{theo:Stability_implies_equilibrium}
does not imply that behavior will necessarily stabilize in an SMDP.
In fact, it is well known from the theory of Markov chains that, even
if no decisions affect the relevant transitions, outcomes need not
stabilize without further assumptions. So one cannot hope to have
general statements regarding convergence of outcomes---this is also
true, for example, in the related context of learning to play Nash
equilibrium in games.\footnote{For example, in the game-theory literature, general global convergence
results have only been obtained in special classes of games\textendash e.g.
zero-sum, potential, and supermodular games (\citealp{hofbauer2002global}).} Thus, the theorem leaves open the question of convergence in specific
settings, a question that requires other tools (e.g., stochastic approximation)
and is best tackled by explicitly studying the dynamics of specific
classes of environments (see the references in the introduction).

\emph{Mixed strategies}: Theorem \ref{theo:Stability_implies_equilibrium}
also raises the question of how a mixed strategy could ever become
stable, given that, in general it is unlikely that agents will hold
beliefs that make them exactly indifferent at any point in time. \cite{fudenberg1993learning}
asked the same question in the context of learning to play mixed strategy
\emph{Nash} equilibria, and answered it by adding small payoff perturbations
a la \cite{harsanyi1973games}: Agents do not actually mix; instead,
every period their payoffs are subject to small perturbations, and
what we call the mixed strategy is simply the probability distribution
generated by playing \emph{pure} strategies and integrating over the
payoff perturbations. We followed this approach in the paper that
introduced Berk-Nash equilibrium in static contexts (Esponda and Pouzo,
2016). The same idea applies here, but we omit payoff perturbations
to reduce the notational burden.\footnote{\cite{doraszelski2010theory} incorporate payoff perturbations in
a dynamic environment.}
\end{rem}

\section{\label{sec:Equilibrium-refinements}Equilibrium refinements}

Theorem \ref{theo:Stability_implies_equilibrium} implies that, for
sufficiently patient players, we should be interested in the following
refinement of Berk-Nash equilibrium.

\medskip{}

\begin{defn}
\label{Def:Berk-Nash-1}A strategy and probability distribution $(\sigma,m)\in\Sigma\times\Delta(Gr(\Gamma))$
is a \textbf{Berk-Nash} \textbf{equilibrium with exhaustive learning}
of the SMDP if it is a Berk-Nash equilibrium that is supported by
a belief $\mu^{*}\in\Delta(\Theta)$ such that, for all $(s,x)\in Gr(\Gamma)$,
\[
\mu^{*}=B(s,x,s',\mu^{*})
\]
for all $s'\in\mathbb{S}$ such that $\bar{Q}_{\mu^{*}}(s'\mid s,x)>0$.%

\end{defn}
\medskip{}

In an equilibrium with exhaustive learning, there is a supporting
belief that is perceived to be a fixed point of the Bayesian operator,
implying that no matter what state and strategy the agent contemplates,
she does not \emph{expect} her belief to change. The requirement
of exhaustive learning does not imply robustness to experimentation.
For example, in the monopoly problem studied in Section \ref{sub:Monopolist},
choosing low price with probability 1 is an equilibrium with exhausted
learning which is supported by the belief that, with probability 1,
$\theta_{L}^{*}=0$. We rule out equilibria that are not robust to
experimentation by introducing a further refinement.\medskip{}

\begin{defn}
An $\varepsilon$-perturbed SMDP is an SMDP wherein strategies are
restricted to belong to 
\[
\Sigma^{\varepsilon}=\left\{ \sigma\in\Sigma:\sigma(x\mid s)\geq\varepsilon\,\,\mbox{for all \ensuremath{(s,x)\in Gr(\Gamma)}}\right\} .
\]

\end{defn}

\medskip{}

\begin{defn}
A strategy and probability distribution $(\sigma,m)\in\Sigma\times\Delta(Gr(\Gamma))$
is a \textbf{perfect Berk-Nash} \textbf{equilibrium} of an SMDP if
there exists a sequence $(\sigma^{\varepsilon},m^{\varepsilon})_{\varepsilon>0}$
of Berk-Nash equilibria with exhaustive learning of the $\varepsilon$-perturbed
SMDP that converges to $(\sigma,m)$ as $\varepsilon\rightarrow0$.\footnote{Formally, in order to have a sequence, we take $\varepsilon>0$ to
belong to the rational numbers; hereinafter we leave this implicit
to ease the notational burden.} \medskip{}

\end{defn}
\cite{selten1975reexamination} introduced the idea of perfection
in extensive-form games. By itself, however, perfection does not guarantee
that all $(s,x)\in Gr(\Gamma)$ are reached in an MDP. The next property
guarantees that all states can be reached when the agent chooses all
strategies with positive probability.\medskip{}

\begin{defn}
\label{de:Fullcomm}An MDP($Q$) satisfies \textbf{full communication}
if, for all $s_{0},s'\in\mathbb{S}$, there exist finite sequences
$(s_{1},...,s_{n})$ and $(x_{0},x_{1},...,x_{n})$ such that $(s_{i},x_{i})\in Gr(\Gamma)$
for all $i=0,1,...,n$ and 
\[
Q(s'\mid s_{n},x_{n})Q(s_{n}\mid s_{n-1},x_{n-1})...Q(s_{1}\mid s_{0},x_{0})>0.
\]
An SMDP satisfies full communication if the corresponding MDP satisfies
it.
\end{defn}
\medskip{}

Full communication is standard in the theory of MDPs and holds in
all of the examples in Section \ref{sec:Examples}.  It guarantees
that there is a single recurrent class of states for all $\varepsilon$-perturbed
environments. In cases where it does not hold and there is more than
one recurrent class of states, one can still apply the following results
by focusing on one of the recurrent classes and ignoring the rest
as long as the agent correctly believes that she cannot go from one
recurrent class to the other.

Full communication guarantees that there are no off-equilibrium outcomes
in a perturbed SMDP. It does not, however, rule out the desire for
experimentation on the equilibrium path. We rule out the latter by
requiring weak identification.

\medskip{}

\begin{prop}
\label{prop:refinement}Suppose that an SMDP is weakly identified,
$\varepsilon$-perturbed, and satisfies full communication. 

(i) If the SMDP is regular and if $(\sigma,m)$ is stable for the
Bayesian-SMDP, it is also stable with exhaustive learning.

(ii) If $(\sigma,m)$ is a Berk-Nash equilibrium, it is also a Berk-Nash
equilibrium with exhaustive learning. \end{prop}
\begin{proof}
See the Appendix.
\end{proof}
\medskip{}

Proposition \ref{prop:refinement} provides conditions such that a
steady state satisfies exhaustive learning and a Berk-Nash equilibrium
can be supported by a belief that satisfies the exhaustive learning
condition. Under these conditions, we can find equilibria that are
robust to experimentation, i.e., perfect equilibria, by considering
perturbed environments and taking the perturbations to zero (see the
examples in Section \ref{sec:Examples}).

The next proposition shows that perfect Berk-Nash is a refinement
of Berk-Nash with exhaustive learning. As illustrated by the monopoly
example in Section \ref{sub:Monopolist}, it is a strict refinement.\medskip{}

\begin{prop}
\label{prop:refinements2}Any perfect Berk-Nash equilibrium of a regular
SMDP is a Berk-Nash equilibrium with exhaustive learning.\end{prop}
\begin{proof}
See the Appendix.
\end{proof}
\medskip{}

We conclude by showing existence of perfect Berk-Nash equilibrium
(hence, of Berk-Nash equilibrium with exhaustive learning, by Proposition
\ref{prop:refinements2}).\medskip{}

\begin{thm}
\label{thm:exist-perfectBN}For any regular SMDP that is weakly identified
and satisfies full communication, there exists a perfect Berk-Nash
equilibrium.\end{thm}
\begin{proof}
See the Appendix.
\end{proof}

\section{Conclusion}

We studied Markov decision processes where the agent has a prior over
a set of possible transition probability functions and updates her
beliefs using Bayes' rule. This problem is relevant in many economic
settings but usually not amenable to analysis. We propose to make
it more tractable by studying asymptotic beliefs and behavior. The
answer to the question ``Can the steady state of a Bayesian-SMDP
be characterized by reference to an MDP with fixed beliefs?'' is
a qualified yes. If the agent is sufficiently impatient, it suffices
to focus on the set of Berk-Nash equilibria. If, on the other hand,
the agent is sufficiently patient and we are interested in steady
states with exhaustive learning, then these steady states are characterized
by the notion of Berk-Nash equilibrium with exhaustive learning. Finally,
if we are interested on equilibria that are robust to experimentation,
we can restrict attention to the set of perfect Berk-Nash equilibria.

Our results hold for both the correctly-specified and misspecified
cases, and we are not aware of any prior \emph{general} results for
either of these cases. For the correctly-specified case, our results
can justify the common assumption in the literature that the agent
knows the transition probability function provided that strong identification
holds (or that there is weak identification and one is interested
in equilibria that are robust to experimentation). In the misspecified
case, our results significantly expand the range of possible applications.

\addcontentsline{toc}{section}{References}

\bibliographystyle{aer}
\bibliography{bibtex}

\appendix

\section*{Appendix\label{sec:Appendix}}

\renewcommand{\baselinestretch}{1.25}
\addcontentsline{toc}{section}{Appendix}

\textbf{Proof of Lemma \ref{Lemma:I(sigma)}. }$I_{Q}(\sigma)$ \emph{is
nonempty}: $M_{\sigma,Q}$ is a linear (hence continuous) self-map
on a convex and compact subset of an Euclidean space (the set of probability
distributions over the finite set $Gr(\Gamma)$); hence, Brower's
theorem implies existence of a fixed point.

$I_{Q}(\sigma)$ \emph{is convex valued}: For all $\alpha\in[0,1]$
and $m_{1},m_{2}\in\Delta(Gr(\Gamma))$, $\alpha M_{\sigma,Q}[m_{1}]+(1-\alpha)M_{\sigma,Q}[m_{2}]=M_{\sigma,Q}[\alpha m_{1}+(1-\alpha)m_{2}]$.
Thus, if $m_{1}=M_{\sigma,Q}[m_{1}]$ and $m_{2}=M_{\sigma,Q}[m_{2}]$,
then $\alpha m_{1}+(1-\alpha)m_{2}=M_{\sigma,Q}[\alpha m_{1}+(1-\alpha)m_{2}]$.

$I_{Q}(\sigma)$ \emph{is upper hemicontinuous and compact valued}:
Fix any sequence $(\sigma_{n},m_{n})_{n}$ in $\Sigma\times\Delta(Gr(\Gamma))$
such that $\lim_{n\rightarrow\infty}(\sigma_{n},m_{n})=(\sigma,m)$
and such that $m_{n}\in I_{Q}(\sigma_{n})$ for all $n$. Since $M_{\sigma_{n},Q}[m_{n}]=m_{n}$,
$||m-M_{\sigma,Q}[m]||\leq||m-m_{n}||+||M_{\sigma_{n},Q}[m_{n}-m]||+||M_{\sigma_{n},Q}[m]-M_{\sigma,Q}[m]||$.
The first term in the RHS vanishes by the hypothesis. The second term
satisfies $||M_{\sigma_{n},Q}[m_{n}-m]||\leq||M_{\sigma_{n},Q}||\times||m_{n}-m||$
and also vanishes.\footnote{For a matrix $A$, $||A||$ is understood as the operator norm.}
For the third term, note that $\sigma\mapsto M_{\sigma,Q}[m]$ is
a linear mapping and $\sup_{\sigma}||M_{\sigma,Q}[m]||\leq\max_{s'}|\sum_{(s,x)\in Gr(\Gamma)}Q(s'\mid s,x)m(s,x)|<\infty$.
Thus $||M_{\sigma_{n},Q}[m]-M_{\sigma,Q}[m]||\leq K\times||\sigma_{n}-\sigma||$
for some $K<\infty$ , and so it also vanishes. Therefore, $m=M_{\sigma,Q}[m]$;
thus, $I_{Q}(\cdot)$ has a closed graph and so $I_{Q}(\sigma)$ is
a closed set. Compactness of $I_{Q}(\sigma)$ follows from compactness
of $\Delta(Gr(\Gamma))$. Therefore, $I_{Q}(\cdot)$ is upper hemicontinuous
(see \cite{aliprantis2006infinite}, Theorem 17.11). $\square$

\medskip{}

The proof of Lemma \ref{Lemma:Theta(m)} relies on the following claim.
The proofs of Claims A, B, and C in this appendix appear in the Online
Appendix .\medskip{}

\textbf{Claim A. (i)} \emph{For any regular SMDP, there exists $\theta^{*}\in\Theta$
and $K<\infty$ such that, for all $m\in\Delta(Gr(\Gamma))$, $K_{Q}(m,\theta^{*})\leq K$.
}\textbf{\emph{(ii)}}\emph{ Fix any $\theta\in\Theta$ and a sequence
$(m_{n})_{n}$ in $\Delta(Gr(\Gamma))$ such that $Q_{\theta}(s'\mid s,x)>0$
for all $(s',s,x)\in\mathbb{S}\times Gr(\Gamma)$ such that $Q(s'\mid s,x)>0$
and $\lim_{n\rightarrow\infty}m_{n}=m$. Then $\lim_{n\rightarrow\infty}K_{Q}(m_{n},\theta)=K_{Q}(m,\theta)$.
}\textbf{\emph{(iii)}}\emph{ $K_{Q}$ is (jointly) lower semicontinuous:
Fix any $(m_{n})_{n}$ and $(\theta_{n})_{n}$ such that $\lim_{n\rightarrow\infty}m_{n}=m$
and $\lim_{n\rightarrow\infty}\theta_{n}=\theta$. Then $\liminf_{n\rightarrow\infty}K_{Q}(m_{n},\theta_{n})\geq K_{Q}(m,\theta)$.}

\medskip{}

\textbf{Proof of Lemma \ref{Lemma:Theta(m)}.} (i) By Jensen's inequality
and strict concavity of $\ln(\cdot)$, $K_{Q}(m,\theta)\geq-\sum_{(s,x)\in Gr(\Gamma)}\ln(E_{Q(\cdot\mid s,x)}[\frac{Q_{\theta}(S'\mid s,x)}{Q(S'\mid s,x)}])m(s,x)=0$,
with equality if and only if $Q_{\theta}(\cdot\mid s,x)=Q_{\theta}(\cdot\mid s,x)$
for all $(s,x)$ such that $m(s,x)>0$.

(ii) \emph{$\Theta_{Q}(m)$ is nonempty}: By Claim A(i), there exists
$K<\infty$ such that the minimizers are in the constraint set $\{\theta\in\Theta:K_{Q}(m,\theta)\leq K\}$.
Because $K_{Q}(m,\cdot)$ is continuous over a compact set, a minimum
exists.

\emph{$\Theta_{Q}(\cdot)$ is uhc} and compact valued: Fix any $(m_{n})_{n}$
and $(\theta_{n})_{n}$ such that $\lim_{n\rightarrow\infty}m_{n}=m$,
$\lim_{n\rightarrow\infty}\theta_{n}=\theta$, and $\theta_{n}\in\Theta_{Q}(m_{n})$
for all $n$. We establish that $\theta\in\Theta_{Q}(m)$ (so that
$\Theta(\cdot)$ has a closed graph and, by compactness of $\Theta$,
it is uhc). Suppose, in order to obtain a contradiction, that $\theta\notin\Theta_{Q}(m)$.
Then, by Claim A(i), there exists $\hat{\theta}\in\Theta$ and $\varepsilon>0$
such that $K_{Q}(m,\hat{\theta})\leq K_{Q}(m,\theta)-3\varepsilon$
and $K_{Q}(m,\hat{\theta})<\infty$. By regularity, there exists $(\hat{\theta}_{j})_{j}$
with $\lim_{j\rightarrow\infty}\hat{\theta}_{j}=\hat{\theta}$ and,
for all $j$, $Q_{\hat{\theta}_{j}}(s'\mid s,x)>0$ for all \emph{$(s',s,x)\in\mathbb{S}^{2}\times\mathbb{X}$
such that $Q(s'\mid s,x)>0$}. We will show that there is an element
of the sequence, $\hat{\theta}_{J}$, that ``does better'' than
$\theta_{n}$ given $m_{n}$, which is a contradiction. Because $K_{Q}(m,\hat{\theta})<\infty$,
continuity of $K_{Q}(m,\cdot)$ implies that there exists $J$ large
enough such that $\left|K_{Q}(m,\hat{\theta}_{J})-K_{Q}(m,\hat{\theta})\right|\leq\varepsilon/2$.
Moreover, Claim A(ii) applied to $\theta=\hat{\theta}_{J}$ implies
that there exists $N_{\varepsilon,J}$ such that, for all $n\geq N_{\varepsilon,J}$,
$\left|K_{Q}(m_{n},\hat{\theta}_{J})-K_{Q}(m,\hat{\theta}_{J})\right|\leq\varepsilon/2$.
Thus, for all $n\geq N_{\varepsilon,J}$, $\bigl|K_{Q}(m_{n},\hat{\theta}_{J})-K_{Q}(m,\hat{\theta})\bigr|\leq\bigl|K_{Q}(m_{n},\hat{\theta}_{J})-K_{Q}(m,\hat{\theta}_{J})\bigr|+\bigl|K_{Q}(m,\hat{\theta}_{J})-K_{Q}(m,\hat{\theta})\bigr|\leq\varepsilon$
and, therefore, 
\begin{equation}
K_{Q}(m_{n},\hat{\theta}_{J})\leq K_{Q}(m,\hat{\theta})+\varepsilon\leq K_{Q}(m,\theta)-2\varepsilon.\label{eq:uhc_1-1}
\end{equation}

Suppose $K_{Q}(m,\theta)<\infty$. By Claim A(iii), there exists $n_{\varepsilon}\geq N_{\varepsilon,J}$
such that $K_{Q}(m_{n_{\varepsilon}},\theta_{n_{\varepsilon}})\geq K_{Q}(m,\theta)-\varepsilon$.
This result, together with (\ref{eq:uhc_1-1}), implies that $K_{Q}(m_{n_{\varepsilon}},\hat{\theta}_{J})\leq K_{Q}(m_{n_{\varepsilon}},\theta_{n_{\varepsilon}})-\varepsilon$.
But this contradicts $\theta_{n_{\varepsilon}}\in\Theta_{Q}(m_{n_{\varepsilon}})$.
Finally, if $K_{Q}(m,\theta)=\infty$, Claim A(iii) implies that there
exists $n_{\varepsilon}\geq N_{\varepsilon,J}$ such that $K_{Q}(m_{n_{\varepsilon}},\theta_{n_{\varepsilon}})\geq2K$,
where $K$ is the bound defined in Claim A(i). But this also contradicts
$\theta_{n_{\varepsilon}}\in\Theta_{Q}(m_{n_{\varepsilon}})$. Thus,
$\Theta_{Q}(\cdot)$ has a closed graph, and so $\Theta_{Q}(m)$ is
a closed set. Compactness of $\Theta_{Q}(m)$ follows from compactness
of $\Theta$. Therefore, $\Theta_{Q}(\cdot)$ is upper hemicontinuous
(see \cite{aliprantis2006infinite}, Theorem 17.11). $\square$

\medskip{}

\textbf{Proof of Theorem \ref{The:Existence}. } Let $\mathbb{W}=\Sigma\times\Delta(Gr(\Gamma))\times\Delta(\Theta)$
and endow it with the product topology (given by the Euclidean one
for $\Sigma\times\Delta(Gr(\Gamma))$ and the weak topology for $\Delta(\Theta)$).
Clearly, $\mathbb{W}\ne\{\emptyset\}$. Since $\Theta$ is compact,
$\Delta(\Theta)$ is compact under the weak topology; $\Sigma$ and
$\Delta(Gr(\Gamma))$ are also compact. Thus by Tychonoff's theorem
(see \cite{aliprantis2006infinite}), $\mathbb{W}$ is compact under
the product topology. $\mathbb{W}$ is also convex. Finally, $\mathbb{W}\subseteq\mathbb{M}^{2}\times rca(\Theta)$
where $\mathbb{M}$ is the space of $|\mathbb{S}|\times|\mathbb{X}|$
real-valued matrices and $rca(\Theta)$ is the space of regular Borel
signed measures endowed with the weak topology. The space $\mathbb{M}^{2}\times rca(\Theta)$
is locally convex with a family of seminorms $\{(\sigma,m,\mu)\mapsto p_{f}(\sigma,m,\mu)=||(\sigma,m)||+|\int_{\Omega}f(x)\mu(dx)|\colon\mbox{\text{ }}f\in\mathbb{C}(\Omega)\}$
($\mathbb{C}(\Omega)$ is the space of real-valued continuous and
bounded functions and $||.||$ is understood as the spectral norm).
Also, we observe that $(\sigma,m,\mu)=0$ iff $p_{f}(\sigma,m,\mu)=0$
for all $f\in\mathbb{C}(\Omega)$, thus $\mathbb{M}^{2}\times rca(\Theta)$
is also Hausdorff.%

Let $\mathcal{T}:\mathbb{W}\rightarrow2^{\mathbb{W}}$ be such that
$\mathcal{T}(\sigma,m,\mu)=\Sigma(\bar{Q}_{\mu})\times I_{Q}(\sigma)\times\Delta(\Theta_{Q}(m))$.
Note that if $(\sigma^{*},m^{*},\mu^{*})$ is a fixed point of $\mathcal{T}$,
then $m^{*}$ is a Berk-Nash equilibrium. By Lemma \ref{Lemma:Sigma(Q)},
$\Sigma(\cdot)$ is nonempty, convex valued, compact valued, and upper
hemicontinuous. Thus, for every $\mu\in\Delta(\Theta)$, $\Sigma(\bar{Q}_{\mu})$
is nonempty, convex valued, and compact valued. Also, because $Q_{\theta}$
is continuous in $\theta$ (by regularity assumption), then $\bar{Q}_{\mu}$
is continuous (under the weak topology) in $\mu$. Since $Q\mapsto\Sigma(Q)$
is upper hemicontinuous, then $\Sigma(\bar{Q}_{\mu})$ is also upper
hemicontinuous as a function of $\mu$. By Lemma \ref{Lemma:I(sigma)},
$I_{Q}(\cdot)$ is nonempty, convex valued, compact valued and upper
hemicontinuous. By Lemma \ref{Lemma:Theta(m)} and the regularity
condition, the correspondence $\Theta_{Q}(\cdot)$ is nonempty, compact
valued, and upper hemicontinuous; hence, the correspondence $\Delta(\Theta_{Q}(\cdot))$
is nonempty, upper hemicontinuous (see \cite{aliprantis2006infinite},
Theorem 17.13), compact valued (see \cite{aliprantis2006infinite},
Theorem 15.11) and, trivially, convex valued. Thus, the correspondence
$\mathcal{T}$ is nonempty, convex valued, compact valued (by Tychonoff's
Theorem), and upper hemicontinuous (see \cite{aliprantis2006infinite},
Theorem 17.28) under the product topology; hence, it has a closed
graph (see \cite{aliprantis2006infinite}, Theorem 17.11). Since
$\mathbb{W}$ is a nonempty compact convex subset of a locally Hausdorff
space, then there exists a fixed point of $\mathcal{T}$ by the Kakutani-Fan-Glicksberg
theorem (see \cite{aliprantis2006infinite}, Corollary 17.55). $\square$

\medskip{}

For the proof of Lemma \ref{Lemma:Berk}, we rely on the following
definitions and Claim. Let $K^{\ast}(m)=\inf_{\theta\in\Theta}K_{Q}(m,\theta)$
and let $\hat{\Theta}\subseteq\Theta$ be a dense set such that, for
all $\theta\in\hat{\Theta}$, $Q_{\theta}(s'\mid s,x)>0$ for all
$(s',s,x)\in\mathbb{S}\times Gr(\Gamma)$ such that $Q(s'\mid s,x)>0$.
Existence of such a set $\hat{\Theta}$ follows from the regularity
assumption.

\medskip{}

\textbf{Claim B.} Suppose $\lim_{t\rightarrow\infty}\left\Vert m_{t}-m\right\Vert =0$
a.s.-$\mathbf{P}^{f}$ . Then: \textbf{(i)} For all $\theta\in\hat{\Theta}$,
\begin{align*}
\lim_{t\rightarrow\infty}t^{-1}\sum_{\tau=1}^{t}\log\frac{Q(s_{\tau}|s_{\tau-1},x_{\tau-1})}{Q_{\theta}(s_{\tau}|s_{\tau-1},x_{\tau-1})}=\sum_{(s,x)\in Gr(\Gamma)}E_{Q(\cdot|s,x)}\Bigl[\log\frac{Q(S'|s,x)}{Q_{\theta}(S'|s,x)}\Bigr]m(s,x)
\end{align*}
a.s.-$\mathbf{P}^{f}$. \textbf{(ii)} For $\mathbf{P}^{f}$-almost
all $h^{\infty}\in\mathbb{H}^{\infty}$ and for any $\epsilon>0$
and $\alpha=(\inf_{\Theta\colon d_{m}(\theta)\geq\epsilon}K_{Q}(m,\theta)-K^{\ast}(m))/3$,
there exists $T$ such that, for all $t\geq T$, 
\begin{align*}
t^{-1}\sum_{\tau=1}^{t}\log\frac{Q(s_{\tau}|s_{\tau-1},x_{\tau-1})}{Q_{\theta}(s_{\tau}|s_{\tau-1},x_{\tau-1})}\geq K^{\ast}(m)+\frac{3}{2}\alpha
\end{align*}
for all $\theta\in\{\Theta\colon d_{m}(\theta)\geq\epsilon\}$, where
$d_{m}(\theta)=\inf_{\tilde{\theta}\in\Theta_{Q}(m)}||\theta-\tilde{\theta}||$.

\medskip{}

\textbf{Proof of Lemma \ref{Lemma:Berk}. }It suffices to show that
$\lim_{t\rightarrow\infty}\int_{\Theta}d_{m}(\theta)\mu_{t}(d\theta)=0$
a.s.-$\mathbf{P}^{f}$ over $\mathcal{H}$. Let $K^{\ast}(m)\equiv K_{Q}(m,\Theta_{Q}(m))$.
For any $\eta>0$ let $\Theta_{\eta}(m)=\left\{ \theta\in\Theta:d_{m}(\theta)<\eta\right\} $,
and $\hat{\Theta}_{\eta}(m)=\hat{\Theta}\cap\Theta_{\eta}(m)$ (the
set $\hat{\Theta}$ is defined in condition 3 of Definition \ref{Def:regular},
i.e., regularity). We now show that $\mu_{0}(\hat{\Theta}_{\eta}(m))>0$.
By Lemma \ref{Lemma:Theta(m)}, $\Theta_{Q}(m)$ is nonempty. By denseness
of $\hat{\Theta}$, $\hat{\Theta}_{\eta}(m)$ is nonempty. Nonemptiness
and continuity of $\theta\mapsto Q_{\theta}$, imply that there exists
a non-empty open set $U\subseteq\hat{\Theta}_{\eta}(m)$. By full
support, $\mu_{0}(\hat{\Theta}_{\eta}(m))>0$.%
{}  Also, observe that for any $\epsilon>0$, $\{\Theta\colon d_{m}(\theta)\geq\epsilon\}$
is compact. This follows from compactness of $\Theta$ and continuity
of $\theta\mapsto d_{m}(\theta)$ (which follows by Lemma \ref{Lemma:Theta(m)}
and an application of the Theorem of the Maximum). Compactness of
$\{\Theta\colon d_{m}(\theta)\geq\epsilon\}$ and lower semi-continuity
of $\theta\mapsto K_{Q}(m,\theta)$ (see Claim A(iii)) imply that
$\inf_{\Theta\colon d_{m}(\theta)\geq\epsilon}K_{Q}(m,\theta)=\min_{\Theta\colon d_{m}(\theta)\geq\epsilon}K_{Q}(m,\theta)>K^{\ast}(m)$.
Let $\alpha\equiv(\min_{\Theta\colon d_{m}(\theta)\geq\epsilon}K_{Q}(m,\theta)-K^{\ast}(m))/3>0$.
Also, let $\eta>0$ be chosen such that $K_{Q}(m,\theta)\leq K^{\ast}(m)+0.25\alpha$
for all $\theta\in\Theta_{\eta}(m)$ (such $\eta$ always exists by
continuity of $\theta\mapsto K_{Q}(m,\theta)$). 

Let $\mathcal{H}_{1}$ be the subset of $\mathcal{H}$ for which the
statements in Claim B hold; note that $\mathbf{P}^{f}\left(\mathcal{H}\setminus\mathcal{H}_{1}\right)=0$.
Henceforth, fix $h^{\infty}\in\mathcal{H}_{1}$; we omit $h^{\infty}$
from the notation to ease the notational burden. By simple algebra
and the fact that $d_{m}$ is bounded in $\Theta$, it follows that,
for all $\epsilon>0$ and some finite $C>0$,
\begin{align*}
\int_{\Theta}d_{m}(\theta)\mu_{t}(d\theta) & =\frac{\int_{\Theta}d_{m}(\theta)Q_{\theta}(s_{t}\mid s_{t-1},x_{t-1})\mu_{t-1}(d\theta)}{\int_{\Theta}Q_{\theta}(s_{t}\mid s_{t-1},x_{t-1})\mu_{t-1}(d\theta)}=\frac{\int_{\Theta}d_{m}(\theta)Z_{t}(\theta)\mu_{0}(d\theta)}{\int_{\Theta}Z_{t}(\theta)\mu_{0}(d\theta)}\\
 & \leq\epsilon+C\frac{\int_{\left\{ \Theta\colon d_{m}(\theta)\geq\epsilon\right\} }Z_{t}(\theta)\mu_{0}(d\theta)}{\int_{\hat{\Theta}_{\eta}(m)}Z_{t}(\theta)\mu_{0}(d\theta)}\equiv\epsilon+C\frac{A_{t}(\epsilon)}{B_{t}(\eta)}.
\end{align*}
where $Z_{t}(\theta)\equiv\prod_{\tau=1}^{t}\frac{Q_{\theta}(s_{\tau}|s_{\tau-1},x_{\tau-1})}{Q(s_{\tau}|s_{\tau-1},x_{\tau-1})}=\exp\left\{ -\sum_{\tau=1}^{t}\log\left(\frac{Q(s_{\tau}|s_{\tau-1},x_{\tau-1})}{Q_{\theta}(s_{\tau}|s_{\tau-1},x_{\tau-1})}\right)\right\} $.
Hence, it suffices to show that 
\begin{equation}
\limsup_{t\rightarrow\infty}\left\{ \exp\left\{ t\left(K^{\ast}(m)+0.5\alpha\right)\right\} A_{t}(\epsilon)\right\} =0\label{eq:A_limit-1-1-1}
\end{equation}
and 
\begin{equation}
\liminf_{t\rightarrow\infty}\left\{ \exp\left\{ t\left(K^{\ast}(m)+0.5\alpha\right)\right\} B_{t}(\eta)\right\} =\infty.\label{eq:B_limit-1-1-1}
\end{equation}

Regarding equation (\ref{eq:A_limit-1-1-1}), we first show that 
\[
\lim_{t\rightarrow\infty}\sup_{\{\Theta\colon d_{m}(\theta)\geq\epsilon\}}\Bigl\{\left(K^{\ast}(m)+0.5\alpha\right)-t^{-1}\sum_{\tau=1}^{t}\log\frac{Q(s_{\tau}|s_{\tau-1},x_{\tau-1})}{Q_{\theta}(s_{\tau}|s_{\tau-1},x_{\tau-1})}\Bigr\}\leq const<0.
\]
To show this, note that, by Claim B(ii) there exists a $T$, such
that for all $t\geq T$, $t^{-1}\sum_{\tau=1}^{t}\log\frac{Q(s_{\tau}|s_{\tau-1},x_{\tau-1})}{Q_{\theta}(s_{\tau}|s_{\tau-1},x_{\tau-1})}\geq K^{\ast}(m)+\frac{3}{2}\alpha$,
for all $\theta\in\{\Theta\colon d_{m}(\theta)\geq\epsilon\}$. Thus,
\begin{align*}
\lim_{t\rightarrow\infty}\sup_{\{\Theta\colon d_{m}(\theta)\geq\epsilon\}}\Bigl\{ K^{\ast}(m)+\frac{\alpha}{2}-t^{-1}\sum_{\tau=1}^{t}\log\frac{Q(s_{\tau}|s_{\tau-1},x_{\tau-1})}{Q_{\theta}(s_{\tau}|s_{\tau-1},x_{\tau-1})}\Bigr\}\leq-\alpha.
\end{align*}
 Therefore,
\begin{align*}
 & \limsup_{t\rightarrow\infty}\left\{ \exp\left\{ t\left(K^{\ast}(m)+0.5\alpha\right)\right\} A_{t}(\epsilon)\right\} \\
 & \leq\limsup_{t\rightarrow\infty}\sup_{\{\Theta\colon d_{m}(\theta)\geq\epsilon\}}\exp\Bigl\{ t\Bigl(\left(K^{\ast}(m)+0.5\alpha\right)-t^{-1}\sum_{\tau=1}^{t}\log\frac{Q(s_{\tau}|s_{\tau-1},x_{\tau-1})}{Q_{\theta}(s_{\tau}|s_{\tau-1},x_{\tau-1})}\Bigr)\Bigr\}\\
 & =0.
\end{align*}

Regarding equation (\ref{eq:B_limit-1-1-1}), by Fatou's lemma and
some algebra it suffices to show that 
\[
\liminf_{t\rightarrow\infty}\exp\left\{ t\left(K^{\ast}(m)+0.5\alpha\right)\right\} Z_{t}(\theta)=\infty>0
\]
(pointwise on $\theta\in\hat{\Theta}_{\eta}(m)$), or, equivalently,
\[
\liminf_{t\rightarrow\infty}\Bigl(K^{\ast}(m)+0.5\alpha-t^{-1}\sum_{\tau=1}^{t}\log\frac{Q(s_{\tau}|s_{\tau-1},x_{\tau-1})}{Q_{\theta}(s_{\tau}|s_{\tau-1},x_{\tau-1})}\Bigr)>0.
\]
By Claim B(i), 
\[
\liminf_{t\rightarrow\infty}\Bigl(K^{\ast}(m)+0.5\alpha-t^{-1}\sum_{\tau=1}^{t}\log\frac{Q(s_{\tau}|s_{\tau-1},x_{\tau-1})}{Q_{\theta}(s_{\tau}|s_{\tau-1},x_{\tau-1})}\Bigr)=K^{\ast}(m)+0.5\alpha-K_{Q}(m,\theta)
\]
(pointwise on $\theta\in\hat{\Theta}_{\eta}(m)$). By our choice of
$\eta$, the RHS is greater than $0.25\alpha$ and our desired result
follows. $\square$

\medskip{}

\textbf{Proof of Theorem \ref{theo:Stability_implies_equilibrium}.
}%
For any $s\in\mathbb{S}$ and $\mu\in\Delta(\Theta)$, let 
\begin{align*}
x(s,\mu) & \equiv\arg\max_{x\in\Gamma(s)}E_{\bar{Q}_{\mu}(\cdot\mid s,x)}\left[\pi(s,x,S')\right]\\
\tilde{\delta}(s,\mu) & \equiv\min_{x\in\Gamma(s)\setminus x(s,\mu)}\Bigl\{\max_{x\in\Gamma(s)}E_{\bar{Q}_{\mu}(\cdot\mid s,x)}\left[\pi(s,x,S')\right]-E_{\bar{Q}_{\mu}(\cdot\mid s,x)}\left[\pi(s,x,S')\right]\Bigr\}\\
\hat{\delta} & \equiv\max\bigl\{\min_{s,\mu}\tilde{\delta}(s,\mu),0\bigr\}\\
\bar{\delta} & \equiv\max\bigl\{\delta\geq0\mid\hat{\delta}-2\frac{\delta}{1-\delta}M\geq0\bigr\}=\frac{\hat{\delta}/M}{2+\hat{\delta}/M},
\end{align*}
where $M\equiv\max_{(s,x)\in Gr(\Gamma),s\in\mathbb{S}'}|\pi(s,x,s')|$.

By Lemma \ref{Lemma:Berk}, for all open sets $U\supseteq\Theta_{Q}(m)$,
$\lim_{t\rightarrow\infty}\mu_{t}\left(U\right)=1$ a.s.-$\mathbf{P}^{f}$
in $\mathcal{H}$. Also Let $g_{\tau}(h^{\infty})(s,x)=\mathbf{1}_{(s,x)}(s_{\tau},x_{\tau})-M_{\sigma_{\tau}}(s,x\mid s_{\tau-1},x_{\tau-1})$
for any $\tau$ and $(s,x)\in Gr(\Gamma)$ and $h^{\infty}\in\mathbb{H}$.
The sequence $(g_{\tau})_{\tau}$ is a Martingale difference and by
analogous arguments to those in the proof of Claim B: $\lim_{t\rightarrow\infty}||t^{-1}\sum_{\tau=0}^{t}g_{\tau}(h^{\infty})||=0$
a.s.-$\mathbf{P}^{f}$. Let $\mathcal{H}^{\ast}$ be the set in $\mathcal{H}$
such that for all $h^{\infty}\in\mathcal{H}^{\ast}$ the following
holds: for all open sets $U\supseteq\Theta_{Q}(m)$, $\lim_{t\rightarrow\infty}\mu_{t}\left(U\right)=1$
and $\lim_{t\rightarrow\infty}||t^{-1}\sum_{\tau=0}^{t}g_{\tau}(h^{\infty})||=0$.
Note that $\mathbf{P}^{f}\left(\mathcal{H}\setminus\mathcal{H}^{\ast}\right)=0$.
Henceforth, fix an $h^{\infty}\in\mathcal{H}^{\ast}$, which we omit
from the notation. 

We first establish that $m\in I_{Q}(\sigma)$. Note that 
\[
\left\Vert m-M_{\sigma,Q}\left[m\right]\right\Vert \leq\left\Vert m-m_{t}\right\Vert +\left\Vert m_{t}-M_{\sigma,Q}\left[m\right]\right\Vert 
\]
where $(s,x)\mapsto M_{\sigma,Q}[p](s,x)\equiv\sum_{\tilde{s},\tilde{x}\in Gr(\Gamma)}M_{\sigma}(s,x|\tilde{s},\tilde{x})p(\tilde{s},\tilde{x})$
for any $p\in\Delta(Gr(\Gamma))$. By stability, the first term in
the RHS vanishes, so it suffices to show that $\lim_{t\rightarrow\infty}||m_{t}-M_{\sigma,Q}\left[m\right]||=0$.
The fact that $\lim_{t\rightarrow\infty}||t^{-1}\sum_{\tau=0}^{t}g_{\tau}||=0$
and the triangle inequality imply
\begin{align}
\lim_{t\rightarrow\infty}\bigl\Vert m_{t}-M_{\sigma,Q}\left[m\right]\bigr\Vert\leq & \lim_{t\rightarrow\infty}\bigl\Vert m_{t}-t^{-1}\sum_{\tau=1}^{t}M_{\sigma_{\tau},Q}(\cdot,\cdot\mid s_{\tau-1},x_{\tau-1})\bigr\Vert\nonumber \\
 & +\lim_{t\rightarrow\infty}\bigl\Vert t^{-1}\sum_{\tau=1}^{t}M_{\sigma_{\tau},Q}(\cdot,\cdot\mid s_{\tau-1},x_{\tau-1})-M_{\sigma,Q}\left[m\right]\bigr\Vert\nonumber \\
= & \lim_{t\rightarrow\infty}\bigl\Vert t^{-1}\sum_{\tau=1}^{t}g_{\tau}\bigr\Vert+\lim_{t\rightarrow\infty}\bigl\Vert t^{-1}\sum_{\tau=1}^{t}M_{\sigma_{\tau},Q}(\cdot,\cdot\mid s_{\tau-1},x_{\tau-1})-M_{\sigma,Q}\left[m\right]\bigr\Vert\nonumber \\
\leq & \lim_{t\rightarrow\infty}\bigl\Vert t^{-1}\sum_{\tau=1}^{t}M_{\sigma_{\tau},Q}(\cdot,\cdot\mid s_{\tau-1},x_{\tau-1})-M_{\sigma,Q}\bigl[t^{-1}\sum_{\tau=1}^{t}\mathbf{1}_{(\cdot,\cdot)}(s_{\tau-1},x_{\tau-1})\bigr]\bigr\Vert\nonumber \\
 & +\lim_{t\rightarrow\infty}\bigl\Vert M_{\sigma,Q}\bigl[t^{-1}\sum_{\tau=1}^{t}\mathbf{1}_{(\cdot,\cdot)}(s_{\tau-1},x_{\tau-1})\bigr]-M_{\sigma,Q}[m]\bigr\Vert.\label{eq:m-Mm-1}
\end{align}
Moreover, by definition of $M_{\sigma,Q}$ (see equation (\ref{eq:transition_kernel-1})),
for all $(s,x)\in Gr(\Gamma)$, 
\begin{align}
t^{-1}\sum_{\tau=1}^{t}M_{\sigma_{\tau},Q}(s,x\mid s_{\tau-1},x_{\tau-1}) & =\sum_{\tilde{s},\tilde{x}\in Gr(\Gamma)}Q(s|\tilde{s},\tilde{x})t^{-1}\sum_{\tau=1}^{t}\sigma_{\tau}(x|s)\mathbf{1}_{(\tilde{s},\tilde{x})}(s_{\tau-1},x_{\tau-1})\label{eq:M1}\\
M_{\sigma,Q}\bigl[t^{-1}\sum_{\tau=1}^{t}\mathbf{1}_{(\cdot,\cdot)}(s_{\tau-1},x_{\tau-1})\bigr] & =\sum_{\tilde{s},\tilde{x}\in Gr(\Gamma)}Q(s\mid\tilde{s},\tilde{x})t^{-1}\sum_{\tau=1}^{t}\sigma(x\mid s)\mathbf{1}_{(\tilde{s},\tilde{x})}(s_{\tau-1},x_{\tau-1}).\label{eq:M2}
\end{align}
Equations (\ref{eq:M1}) and (\ref{eq:M2}) and stability ($\sigma_{t}\rightarrow\sigma)$
imply that the first term in the RHS of \ref{eq:m-Mm-1} vanishes.
The second term in the RHS also vanishes under stability due to continuity
of the operator $M_{\sigma}[.]$ and the fact that $t^{-1}\sum_{\tau=1}^{t}\mathbf{1}_{(\cdot,\cdot)}(s_{\tau-1},x_{\tau-1})=\frac{t-1}{t}m_{t-1}\left(\cdot,\cdot\right)$.
Thus, $\left\Vert m-M_{\sigma,Q}\left[m\right]\right\Vert =0$, and
so $m\in I_{Q}(\sigma)$. 

Therefore, for proving cases (i) and (ii), we need to establish that,
for each case, there exists $\mu\in\Delta(\Theta_{Q}(m))$ such that
$\sigma$ is an optimal strategy for the MDP($\bar{Q}_{\mu}$).

(i) Consider any $\delta\in[0,\bar{\delta}]$. Since $\Delta(\Theta)$
is compact under the weak topology, there exists a subsequence of
$(\mu_{t})_{t}$ --- which we still denote as $(\mu_{t})_{t}$ ---
such that $\mu_{t}\overset{w}{\rightarrow}\mu_{\infty}$ and $\mu_{\infty}\in\Delta(\Theta_{Q}(m))$.
Since $\sigma_{t}\in\bar{\Sigma}(\mu_{t})$ for all $t$ and $\bar{\Sigma}$
is uhc (see Lemma \ref{Lemma:Sigma(Q)-1}), stability $(\sigma_{t}\rightarrow\sigma)$
implies $\sigma\in\bar{\Sigma}(\mu_{\infty})$. We conclude by showing
that $\sigma$ is an optimal strategy for the MDP($\bar{Q}_{\mu_{\infty}}$).
If $\delta=\bar{\delta}=0$, this assertion is trivial. If $\bar{\delta}\geq\delta>0$,
it suffices to show that 
\begin{align}
x(s,\mu_{\infty}) & =\arg\max_{x\in\Gamma(s)}\int_{\mathbb{S}}\left\{ \pi(s,x,s')+\delta W(s',B(s,x,s',\mu_{\infty}))\right\} \bar{Q}_{\mu_{\infty}}(ds'|s,x)\nonumber \\
 & =\arg\max_{x\in\Gamma(s)}\int_{\mathbb{S}}\left\{ \pi(s,x,s')+\delta W(s',\mu_{\infty})\right\} \bar{Q}_{\mu_{\infty}}(ds'|s,x).\label{eq:argmax's}
\end{align}
We conclude by establishing (\ref{eq:argmax's}). Note that, since
$\bar{\delta}>0$, it follows that $\hat{\delta}>0$, which in turn
implies that $x(s,\mu_{\infty})$ is a singleton. The first equality
in (\ref{eq:argmax's}) holds because, by definition of $\bar{\delta}$,
\[
E_{\bar{Q}_{\mu_{\infty}}(\cdot|s,x(s,\mu_{\infty}))}\left[\pi(s,x(s,\mu_{\infty}),S')\right]-E_{\bar{Q}_{\mu_{\infty}}(\cdot|s,x)}\left[\pi(s,x,S')\right]\geq\hat{\delta}\geq2\frac{\delta M}{1-\delta}>0
\]
for all $x\in\Gamma(s)\backslash\{x(s,\mu_{\infty})\}$, and, by definition
of $M$,
\[
2\frac{\delta M}{1-\delta}\geq\delta\int_{\mathbb{S}}\Bigl\{ W(s',B(s,x,s',\mu_{\infty}))\bar{Q}_{\mu_{\infty}}(ds'|s,x)-W(s',B(s,x(s,\mu_{\infty}),s',\mu_{\infty}))\bar{Q}_{\mu_{\infty}}(ds'|s,x(s,\mu_{\infty}))\Bigr\}.
\]
The second equality in (\ref{eq:argmax's}) holds by similar arguments.

(ii) By stability with exhaustive learning, there exists a subsequence
$(\mu_{t(j)})_{j}$ such that $\mu_{t(j)}\overset{w}{\rightarrow}\mu^{*}$.
This fact and the fact that for all open $U\supseteq\Theta_{Q}(m)$,
$\lim_{t\rightarrow\infty}\mu_{t(j)}\left(U\right)=1$, imply that
$\mu^{*}\in\Delta(\Theta_{Q}(m))$. Since $\sigma_{t(j)}\in\bar{\Sigma}(\mu_{t(j)})$
for all $j$ and $\bar{\Sigma}$ is uhc (see Lemma \ref{Lemma:Sigma(Q)-1}),
stability $(\sigma_{t}\rightarrow\sigma)$ implies $\sigma\in\bar{\Sigma}(\mu^{*})$.
Moreover, by condition of stability with exhaustive learning (i.e.,
$\mu^{\ast}=B(s,x,s',\mu^{\ast})$ for all $(s,x)\in Gr(\Gamma)$
and $s'\in supp(\bar{Q}_{\mu^{*}}(\cdot|s,x))$), $W(s,\mu^{*})=\max_{x\in\Gamma(s)}\int_{\mathbb{S}}\left\{ \pi(s,x,s')+\delta W(s',\mu^{*})\right\} \bar{Q}_{\mu^{*}}(ds'|s,x)$
for all $s\in\mathbb{S}$. Then, by uniqueness of the value function,
$\sigma$ is an optimal strategy for the MDP($\bar{Q}_{\mu_{\ast}}$).
$\square$

\medskip{}

The proof of Proposition\textbf{ }\ref{prop:refinement}\textbf{ }relies
on the following claim. \medskip{}

\textbf{Claim C. }If $(\sigma,m)\in\Sigma\times\Delta(\mathbb{S}\times\mathbb{X})$
is such that $\sigma\in\Sigma^{\varepsilon}$ and $m\in I_{Q}(m)$
with $Q$ satisfying the full communication condition in Definition
\ref{de:Fullcomm}, then  $m(s,x)>0$ for all $(s,x)\in Gr(\Gamma)$.

\medskip{}

\textbf{Proof of Proposition \ref{prop:refinement}.} (i) We show
that, if $(\sigma,m)$ is stable for a Bayesian-SMDP that is $\varepsilon$-perturbed,
weakly identified and satisfies full communication (and has a prior
$\mu_{0}$ and policy function $f$ ), then $(\sigma,m)$ is stable
with exhaustive learning. That is, we must find a subsequence $(\mu_{t(j)})_{j}$
such that $\mu_{t(j)}$ converges weakly to $\mu^{\ast}$ and $\mu^{\ast}=B(s,x,s',\mu^{\ast})$
for any $(s,x)\in Gr(\Gamma)$ and $s'\in supp(\bar{Q}_{\mu^{\ast}}(\cdot\mid s,x))$.
By compactness of $\Delta(\Theta)$, there always exists a convergent
subsequence with limit point $\mu^{\ast}\in\Delta(\Theta)$. By Lemma
\ref{Lemma:Berk}, $\mu^{\ast}\in\Delta(\Theta_{Q}(m))$. By assumption,
$\sigma\in\Sigma^{\varepsilon}$ and, by the arguments given in the
proof of Theorem \ref{theo:Stability_implies_equilibrium}, $m\in I_{Q}(\sigma)$.
Since the SMDP satisfies full-communication, by Claim C, $supp(m)=Gr(\Gamma)$.
This result, the fact that $\mu^{\ast}\in\Delta(\Theta_{Q}(m))$,
and weak identification imply strong identification, i.e., for any
$\theta_{1}$ and $\theta_{2}$ in the support of $\mu^{\ast}$, $Q_{\theta_{1}}(\cdot\mid s,x)=Q_{\theta_{2}}(\cdot\mid s,x)$
for all $(s,x)\in Gr(\Gamma)$. Hence, it follows that, for all $A\subseteq\Theta$
Borel and for all $(s,x)\in Gr(\Gamma)$ and $s'\in\mathbb{S}$ such
that $\bar{Q}_{\mu^{\ast}}(s'\mid s,x)>0$ (i.e., $\int_{\Theta}Q_{\theta}(s'\mid s,x)\mu^{\ast}(d\theta)>0$),
\[
B(s,x,s',\mu^{\ast})(A)=\frac{\int_{A}Q_{\theta}(s'\mid s,x)\mu^{\ast}(d\theta)}{\int_{\Theta}Q_{\theta}(s'\mid s,x)\mu^{\ast}(d\theta)}=\mu^{\ast}(A).
\]
Thus, $\mu^{\ast}$ satisfies the desired condition.

(ii) We prove that if $(\sigma,m)$ is a Berk-Nash equilibrium, then
it is also a Berk-Nash equilibrium with exhaustive learning. Let $\mu$
be the supporting equilibrium belief. By Claim C and weak identification,
it follows that there is strong identification, and so for any $\theta_{1}$
and $\theta_{2}$ in the support of $\mu$, $Q_{\theta_{1}}(\cdot\mid s,x)=Q_{\theta_{2}}(\cdot\mid s,x)$
for all $(s,x)\in Gr(\Gamma)$. It follows that, for all $A\subseteq\Theta$
Borel and for all $(s,x)\in Gr(\Gamma)$ and $s'\in\mathbb{S}$ such
that $\bar{Q}_{\mu}(s'\mid s,x)>0$ (i.e., $\int_{\Theta}Q_{\theta}(s'\mid s,x)\mu(d\theta)>0$),
\[
B(s,x,s',\mu)(A)=\frac{\int_{A}Q_{\theta}(s'\mid s,x)\mu(d\theta)}{\int_{\Theta}Q_{\theta}(s'\mid s,x)\mu(d\theta)}=\mu(A).
\]
Thus, $(\sigma,m)$ is a Berk-Nash equilibrium with exhaustive learning.
$\square$\medskip{}

\textbf{Proof of Proposition \ref{prop:refinements2}.} Suppose $(\sigma,m)$
is a perfect Berk-Nash equilibrium and let $(\sigma^{\varepsilon},m^{\varepsilon},\mu^{\varepsilon})_{\varepsilon}$
be the associated sequence of equilibria with exhausted learning such
that $\lim_{\varepsilon\rightarrow0}(\sigma^{\varepsilon},m^{\varepsilon})=(\sigma,m)$.
By possibly going to a sub-sequence, let $\mu=\lim_{\varepsilon\rightarrow0}\mu^{\varepsilon}$
(under the weak topology). By the upper hemicontinuity of the equilibrium
correspondence $\mathcal{T}(\sigma,m,\mu)=\Sigma(\bar{Q}_{\mu})\times I_{Q}(\sigma)\times\Delta(\Theta_{Q}(m))$
(see the proof of Theorem \ref{The:Existence}), $(\sigma,m)$ is
a Berk-Nash equilibrium with supporting belief $\mu$. We conclude
by showing that $(\sigma,m)$ is a Berk-Nash equilibrium with exhaustive
learning.

For all $(s,x)\in Gr(\Gamma)$ and $s'\in supp\left(\bar{Q}_{\mu}(\cdot|s,x)\right)$,
and for all $f:\Theta\rightarrow\mathbb{R}$ bounded and continuous,
$\bigl|\int f(\theta)\mu(d\theta)-\int f(\theta)B(s,x,s',\mu)(d\theta)\bigr|\leq\bigl|\int f(\theta)\mu(d\theta)-\int f(\theta)\mu^{\varepsilon}(d\theta)\bigr|+\bigl|\int f(\theta)\mu^{\varepsilon}(d\theta)-\int f(\theta)B(s,x,s',\mu)(d\theta)\bigr|$.
The first term in the RHS vanishes as $\varepsilon\rightarrow0$ by
definition of weak convergence. For the second term, note that, for
sufficiently small $\varepsilon$, $s'\in supp\left(\bar{Q}_{\mu^{\varepsilon}}(\cdot|s,x)\right)$,
and so, since $\mu^{\varepsilon}=B(s,x,s',\mu^{\varepsilon})$ for
any $(s,x)\in Gr(\Gamma)$ and $s'\in supp\left(\bar{Q}_{\mu^{\varepsilon}}(\cdot|s,x)\right)$,
we can replace $\int f(\theta)\mu^{\varepsilon}(d\theta)$ with $\int f(\theta)B(s,x,s',\mu^{\varepsilon})(d\theta)$.
Thus, the second term vanishes by continuity of the Bayesian operator.
Therefore, by a standard argument\footnote{Suppose $\mu_{1},\mu_{2}$ in $\Delta(\Theta)$ are such that $\left|\int f(\theta)\mu_{1}(d\theta)-\int f(\theta)\mu_{2}(d\theta)\right|=0$
for any $f$ bounded and continuous. Then, for any $F\subseteq\Theta$
closed, $\mu_{1}(F)-\mu_{2}(F)\leq E_{\mu_{1}}[f_{F}(\theta)]-\mu_{2}(F)=E_{\mu_{2}}[f_{F}(\theta)]-\mu_{2}(F)$,
where $f_{F}$ is any continuous and bounded and $f_{F}\geq1_{\{F\}}$;
we call the class of such functions $C_{F}$. Thus, $\mu_{1}(F)-\mu_{2}(F)\leq\inf_{f\in C_{F}}E_{\mu_{2}}[f(\theta)]-\mu_{2}(F)=0$,
where the equality follows from an application of the monotone convergence
theorem. An analogous trick yields the reverse inequality and, therefore,
$\mu_{1}(F)=\mu_{2}(F)$ for any $F\subseteq\Theta$ closed. Borel
measures over $\Theta$ are inner regular (also known as tight; see
\cite{aliprantis2006infinite}, Ch. 12, Theorem 12.7). Thus, for
any Borel set $A\subseteq\Theta$ and any $\epsilon>0$, there exists
a $F\subseteq A$ compact such that $\mu_{i}(A\setminus F)<\epsilon$
for all $i=1,2$. Therefore $\mu_{1}(A)-\mu_{2}(A)\leq\mu_{1}(A)-\mu_{2}(F)\leq\mu_{1}(F)-\mu_{2}(F)+\epsilon$.
By our previous result, it follows that $\mu_{1}(A)-\mu_{2}(A)\leq\epsilon$.
A similar trick yields the reverse inequality and, since $\epsilon$
is arbitrary, this implies that $\mu_{1}(A)=\mu_{2}(A)$ for all $A\subseteq\Theta$
Borel.}, $\mu(A)=B(s,x,s',\mu)(A)$ for all $A\subseteq\Theta$ Borel and
all $(s,x)\in Gr(\Gamma)$ and $s'\in supp\left(\bar{Q}_{\mu}(\cdot|s,x)\right)$,
which implies that $(\sigma,m)$ is a Berk-Nash equilibrium with exhaustive
learning.$\square$

\medskip{}

\textbf{Proof of Theorem \ref{thm:exist-perfectBN}.} Existence of
a Berk-Nash equilibrium of an $\varepsilon$-perturbed environment,
$(\sigma^{\varepsilon},m^{\varepsilon})$, follows for all $\varepsilon\in(0,\bar{\varepsilon}]$,
where $\bar{\varepsilon}=1/(|\mathbb{X}|+1)$, from the same arguments
used to establish existence for the case $\varepsilon=0$ (see Theorem
\ref{The:Existence}). Weak identification, full communication and
Proposition \ref{prop:refinement}(ii) imply that there exists a sequence
$(\sigma^{\varepsilon},m^{\varepsilon})_{\varepsilon>0}$ of Berk-Nash
equilibrium with exhaustive learning. By compactness of $\Sigma\times\Delta(Gr(\Gamma))$,
there is a convergent subsequence, which is a perfect Berk-Nash equilibrium
by definition. $\square$

\newpage{}

\section*{Online Appendix\label{sec:onlineappendix}}

\addcontentsline{toc}{section}{Online Appendix}

\section*{Proof of Lemmas \ref{Lemma:Sigma(Q)} and \ref{Lemma:Sigma(Q)-1}.}

The proof of these lemmas is standard. We only prove Lemma \ref{Lemma:Sigma(Q)};
the proof of Lemma \ref{Lemma:Sigma(Q)-1} is analogous.

\medskip{}

\textbf{Proof of Lemma \ref{Lemma:Sigma(Q)}. }(i) Let $T_{Q}:L^{\infty}(\mathbb{S})\rightarrow L^{\infty}(\mathbb{S})$
be the Bellman operator, $T_{Q}[W](s)=\max_{\hat{x}\in\Gamma(s)}\int_{\mathbb{S}}\left\{ \pi(s,\hat{x},s')+\delta W(s')\right\} Q(ds'|s,\hat{x})$.
By standard arguments, $T_{Q}$ is a contraction with modulus $\delta$,
and so there is a unique fixed point $V_{Q}\in L^{\infty}(\mathbb{S})$.
To establish continuity in $Q$, let $V_{Q_{n}}$ be a sequence of
fixed points given $Q_{n}$ such that $Q_{n}$ converges to $Q$ and
let $V_{Q}$ be the fixed point given $Q$. Then 
\begin{align*}
||V_{Q_{n}}-V_{Q}||_{L^{\infty}} & \leq||T_{Q_{n}}[V_{Q_{n}}]-T_{Q_{n}}[V_{Q}]||_{L^{\infty}}+||T_{Q_{n}}[V_{Q}]-T_{Q}[V_{Q}]||_{L^{\infty}}\\
 & \leq\delta||V_{Q_{n}}-V_{Q}||_{L^{\infty}}+||T_{Q_{n}}[V_{Q}]-T_{Q}[V_{Q}]||_{L^{\infty}}
\end{align*}
and, since $\delta\in[0,1)$, it only remains to show that $||T_{Q_{n}}[V_{Q}]-T_{Q}[V_{Q}]||_{L^{\infty}}\rightarrow0$.
Note that, for any $s\in\mathbb{S}$, 
\[
T_{Q_{n}}[V_{Q}](s)-T_{Q}[V_{Q}](s)\leq\int_{\mathbb{S}}(\pi(s,\hat{x}_{n},s')+\delta V_{Q}(s'))\{Q_{n}(ds'|s,\hat{x}_{n})-Q(ds'|s,\hat{x}_{n})\}
\]
where $\hat{x}_{n}\in\arg\max\int_{\mathbb{S}}\left\{ \pi(s,\hat{x},s')+\delta V_{Q}(s')\right\} Q_{n}(ds'|s,\hat{x})$.
Since $V_{Q}$ and $\pi$ are in $L^{\infty}(\mathbb{S})$ and $|\mathbb{S}|<\infty$,
it follows that $T_{Q_{n}}[V_{Q}](s)-T_{Q}[V_{Q}](s)\leq C||Q_{n}-Q||$
for some finite constant $C$. Using similar arguments, one can show
that $T_{Q}[V_{Q}](s)-T_{Q_{n}}[V_{Q}](s)\leq C||Q_{n}-Q||$. Therefore,
$||T_{Q_{n}}[V_{Q}]-T_{Q}[V_{Q}]||_{L^{\infty}}\leq C||Q_{n}-Q||$
and the desired result follows because $||Q_{n}-Q||\rightarrow0$. 

(ii) For each $s\in\mathbb{S}$ and $Q\in\Delta(\mathbb{S})^{Gr(\Gamma)}$,
let $X_{s}(Q)\equiv\arg\max_{\hat{x}\in\Gamma(s)}U_{s}(\hat{x},Q)$,
where $U_{s}(\hat{x},Q)=\int_{\mathbb{S}}\left\{ \pi(s,\hat{x},s')+\delta V_{Q}(s')\right\} Q(ds'|s,\hat{x})$.
Note that

$\Sigma(Q)=\left\{ \sigma\in\Sigma\colon\forall s\in\mathbb{S},\sigma(\cdot|s)\in\Delta(X_{s}(Q))\right\} $
is isomorphic to $\times_{s\in\mathbb{S}}\Delta(X_{s}(Q))$, in the
sense that $\sigma\in\Sigma(Q)$ iff $\left(\sigma(\cdot|s_{1}),....,\sigma(\cdot|s_{|\mathbb{S}|})\right)\in\times_{s\in\mathbb{S}}\Delta(X_{s}(Q))$.
By part (i), $U_{s}$ is continuous, and so the Theorem of the Maximum
implies that $X_{s}(Q)$ is nonempty, compact-valued, and upper hemicontinuous
in $Q$. By Theorem 17.13 in \cite{aliprantis2006infinite}, $Q\mapsto\Delta(X_{s}(Q))$
is also non-empty, compact-valued and upper hemicontinuous for each
$s\in\mathbb{S}$. By Tychonoff's Theorem, so is $\times_{s\in\mathbb{S}}\Delta(X_{s}(Q))$,
and consequently $\Sigma(Q)$. Finally, to establish convexity of
$\Sigma(Q)$, let $\sigma,\sigma'\in\Sigma(Q)$, $\alpha\in(0,1)$
and $\sigma_{\alpha}=\alpha\sigma+(1-\alpha)\sigma'$. Then, for all
$s\in\mathbb{S}$, $supp\,\sigma_{\alpha}(\cdot\mid s)=supp\,\sigma(\cdot\mid s)\cup supp\,\sigma(\cdot\mid s)\subseteq X_{s}(Q)$,
and so $\sigma_{\alpha}\in\Sigma(Q)$. \textbf{$\square$}

\bigskip{}

\subsection*{Proof of Claims A, B, and C}

\emph{Notation.} For the proofs of Claim A and B, let $\mathbb{Z}=\mathbb{S}\times Gr(\Gamma)$.
For each $z=(s',s,x)\in\mathbb{Z}$ and $m\in\Delta(Gr(\Gamma))$,
define $\bar{P}_{m}(z)=Q(s'\mid s,x)m(s,x)$. We sometimes abuse notation
and write $Q(z)\equiv Q(s'\mid s,x)$, and similarly for $Q_{\theta}$.\bigskip{}

\textbf{Proof of Claim A.} (i) By regularity and finiteness of $\mathbb{Z}$,
there exists $\theta_{*}\in\Theta$ and $\alpha\in(0,1)$ such that
$Q_{\theta^{*}}(z)\geq\alpha$ for all $z\in\mathbb{Z}$ such that
$Q(z)>0$. Thus, for all \emph{$m\in\Delta(Gr(\Gamma))$}, $K_{Q}(m,\theta^{*})\leq-E_{\bar{P}_{m}}[\ln Q_{\theta^{*}}(Z)]\leq-\ln\alpha$. 

(ii) $K_{Q}(m_{n},\theta)-K_{Q}(m,\theta)=\sum_{z:Q(z)>0}(\bar{P}_{m_{n}}(z)-\bar{P}_{m}(z))(\ln Q(z)-\ln Q_{\theta}(z))$.
By the assumption that $Q_{\theta}(z)>0$ for all $z$ such that $Q(z)>0$,
$\left(\ln Q(z)-\ln Q_{\theta}(z)\right)$ is bounded for all $z$
such that $Q(z)>0$. In addition, $\bar{P}_{m_{n}}(z)-\bar{P}_{m}(z)$
converges to zero for all $z\in\mathbb{Z}$ due to linearity of $\bar{P}_{\cdot}$
and due to convergence of $m_{n}$ to $m$.

(iii) $K^{i}(\sigma_{n},\theta_{n}^{i})-K(\sigma,\theta^{i})=\sum_{z:Q(z)>0}(\bar{P}_{m_{n}}(z)-\bar{P}_{m}(z))\ln Q(z)+\sum_{z:Q(z)>0}(\bar{P}_{m}(z)\ln Q_{\theta}(z)-\bar{P}_{m_{n}}(z)\ln Q_{\theta_{n}}(z))$.
The first term in the RHS converges to zero (same argument as Claim
A(ii)). The proof concludes by showing that, for all $z$, 
\begin{equation}
\lim\inf_{n\rightarrow\infty}-\bar{P}_{m_{n}}(z)\ln Q_{\theta_{n}}(z)\geq-\bar{P}_{m}(z)\ln Q_{\theta}(z).\label{eq:lsm}
\end{equation}
Suppose $\lim\inf_{n\rightarrow\infty}-\bar{P}_{m_{n}}(z)\ln Q_{\theta_{n}}(z)\leq M<\infty$
(if not, (\ref{eq:lsm}) holds trivially). Then either (i) $\bar{P}_{m_{n}}(z)\rightarrow\bar{P}_{m}(z)>0$,
in which case (\ref{eq:lsm}) holds with equality by continuity of
$Q_{\theta}(z)$ in $\theta$, or (ii) $\bar{P}_{m_{n}}(z)\rightarrow\bar{P}_{m}(z)=0$,
in which case (\ref{eq:lsm}) holds because its RHS is zero (by convention
that $0\ln0=0$) and its LHS is always nonnegative. $\square$

\bigskip{}

\textbf{Proof of Claim B. }(i) For any $z\in\mathbb{Z}$ and any $h^{\infty}\in\mathcal{H}$,
let $freq_{t}(h^{\infty})(z)\equiv t^{-1}\sum_{\tau=0}^{t-1}\mathbf{1}_{\{z\}}(z_{\tau})$.
Observe that $t^{-1}\sum_{\tau=1}^{t}\log\left(\frac{Q(s_{\tau}|s_{\tau-1},x_{\tau-1})}{Q_{\theta}(s_{\tau}|s_{\tau-1},x_{\tau-1})}\right)=\kappa_{1t}(h^{\infty})+\kappa_{2}-\kappa_{3t}(h^{\infty},\theta)$,
where $\kappa_{1t}(h^{\infty})=\sum_{z\in\mathbb{Z}}\left(freq_{t}(h^{\infty})(z)-\bar{P}_{m}(z)\right)\ln Q(z)$,
$\kappa_{2}=\sum_{z\in\mathbb{Z}\colon Q(z)>0}\bar{P}_{m}(z)\ln Q(z)$,
and $\kappa_{3t}(h^{\infty},\theta)=\sum_{z\in\mathbb{Z}}freq_{t}(h^{\infty})(z)\ln Q_{\theta}(z)$.

We first show that $\lim_{t\rightarrow\infty}\kappa_{1t}(h^{\infty})=0$
a.s.-$\mathbf{P}^{f}$. To do this, let $g_{t}(h^{\infty},z)\equiv\left(\mathbf{1}_{\{z\}}(z_{\tau})-\bar{P}_{m}(z)\right)\ln Q(z)$,
and observe that $(g_{t}(\cdot,z))_{t}$ is a martingale difference
sequence. Let $h^{t}$ denote the partial history until time $t$
and $L_{t}(h^{\infty},z)=\sum_{\tau=1}^{t}\tau^{-1}g_{\tau}(h^{\infty},z)$;
note that $E_{\mathbf{P}^{f}(\cdot\mid h^{t})}\bigl[L_{t+1}(h^{\infty},z)\bigr]=L_{t}(h^{\infty},z)$
and so $(L_{t}(\cdot,z))_{t}$ is a martingale with respect to $\mathbf{P}^{f}$.
Moreover, $E_{\mathbf{P}^{f}(\cdot\mid h^{t})}\left[|g_{t}(h^{\infty},z)|^{2}\right]\leq(\ln Q(z))^{2}Q(z)$,
which is bounded by $1$; this result, the Law of iterated expectations
and the fact that $(g_{t}(\cdot,z))_{t}$ are uncorrelated, imply
that $\sup_{t}E_{\mathbf{P}^{f}}\left[|L_{t}(h^{\infty},z)|^{2}\right]\leq M$
for $M<\infty$. Hence, by the Martingale convergence Theorem (see
Theorem 5.2.8 in \cite{Durrett2010}) $L_{t}(h^{\infty},z)$ converges
a.s.-$\mathbf{P}^{f}$ to a finite $L_{\infty}(h^{\infty},z)$. By
Kronecker's lemma (\cite{pollard2001}, page 105)\footnote{This lemma implies that for a sequence $(\ell_{t})_{t}$ if $\sum_{\tau}\ell_{\tau}<\infty$,
then $\sum_{\tau=1}^{t}\frac{b_{\tau}}{b_{t}}\ell_{\tau}\rightarrow0$
where $(b_{t})_{t}$ is a nondecreasing positive real valued sequence
that diverges to $\infty$. We can apply the lemma with $\ell_{t}\equiv t^{-1}g_{t}$
and $b_{t}=t$. }, $\lim_{t\rightarrow\infty}t^{-1}\sum_{\tau=1}^{t}g_{\tau}(h^{\infty},z)=0$
a.s.-$\mathbf{P}^{f}$, for all (uniformly) $z\in\mathbb{Z}$. Thus,
$\lim_{t\rightarrow\infty}\kappa_{1t}(h^{\infty})=0$ a.s.-$\mathbf{P}^{f}$.

We also note that analogous arguments show that 
\[
\lim_{t\rightarrow\infty}freq_{t}(h^{\infty},z)=\bar{P}_{m}(z)
\]
a.s.-$\mathbf{P}^{f}$, for all (uniformly) $z\in\mathbb{Z}$. 

Since $\theta\in\hat{\Theta}$, $z\mapsto-\log\left(Q_{\theta}(z)\right)$
is bounded. Thus by analogous arguments to those used to show $\lim_{t\rightarrow\infty}\kappa_{1t}(h^{\infty})=0$
a.s.-$\mathbf{P}^{f}$, it follows that, for any $\theta\in\hat{\Theta}$,
$\lim_{t\rightarrow\infty}\kappa_{3t}(h^{\infty},\theta)=\sum_{z\in\mathbb{Z}}\bar{P}_{m}(z)\ln Q_{\theta}(z)$
a.s.-$\mathbf{P}^{f}$. This result and the fact that $\lim_{t\rightarrow\infty}\kappa_{1t}(h^{\infty})=0$
a.s.-$\mathbf{P}^{f}$, imply that $\lim_{t\rightarrow\infty}t^{-1}\sum_{\tau=1}^{t}\log\left(\frac{Q(s_{\tau}|s_{\tau-1},x_{\tau-1})}{Q_{\theta}(s_{\tau}|s_{\tau-1},x_{\tau-1})}\right)=$

$\sum_{z\in\mathbb{Z}}\bar{P}_{m}(z)\log\left(\frac{Q(z)}{Q_{\theta}(z)}\right)=\sum_{(s,x)\in Gr(\Gamma)}E_{Q(\cdot|s,x)}\left[\log\left(\frac{Q(S'|s,x)}{Q_{\theta}(S'|s,x)}\right)\right]m(s,x)$
for any $\theta\in\hat{\Theta}$ a.s.-$\mathbf{P}^{f}$, as desired.

(ii) For any $\xi>0$, define $\Theta_{m,\xi}$ to be the set such
that $\theta\in\Theta_{m,\xi}$ if and only if $Q_{\theta}(z)\geq\xi$
for all $z$ such that $\bar{P}_{m}(z)>0$. For any $\xi>0$, let
$\zeta_{\xi}=-\alpha/(\#\mathbb{Z}4\ln\xi)>0$. By the fact that $\lim_{t\rightarrow\infty}freq_{t}(h^{\infty},z)=\bar{P}_{m}(z)$
a.s.-$\mathbf{P}^{f}$, for all (uniformly) $z\in\mathbb{Z}$, $\exists\hat{t}_{\zeta_{\xi}}$
such that, $\forall t\geq\hat{t}_{\zeta_{\xi}}$, 
\begin{align*}
\kappa_{3t}(h^{\infty},\theta) & \leq\sum_{\{z:\bar{P}_{m}(z)>0\}}freq_{t}(h^{\infty})(z)\ln Q_{\theta}(z)\leq\sum_{\{z^{i}:\bar{P}_{m}(z)>0\}}\left(\bar{P}_{m}(z)-\zeta_{\xi}\right)\ln Q_{\theta}(z)\\
 & \leq\sum_{(s,x)\in Gr(\Gamma)}E_{Q(\cdot\mid s,x)}\left[\ln Q_{\theta}(S'\mid s,x)\right]m(s,x)-\#\mathbb{Z}\zeta_{\xi}\ln\xi,
\end{align*}
 a.s.-$\mathbf{P}^{f}$ and$\forall\theta\in\Theta_{m,\xi}$ (since
$Q_{\theta}(z)\geq\xi$ $\forall z$ such that $\bar{P}_{m}(z)>0$).
The above expression, the fact that \emph{$\alpha/4=-\#\mathbb{Z}\zeta_{\xi}\ln\xi$,}
and the fact that $t^{-1}\sum_{\tau=1}^{t}\log\left(\frac{Q(s_{\tau}|s_{\tau-1},x_{\tau-1})}{Q_{\theta}(s_{\tau}|s_{\tau-1},x_{\tau-1})}\right)=\kappa_{1t}(h^{\infty})+\kappa_{2}-\kappa_{3t}(h^{\infty},\theta)$
imply that $\forall t\geq\hat{t}_{\zeta_{\xi}}$, 
\begin{equation}
t^{-1}\sum_{\tau=1}^{t}\log\left(\frac{Q(s_{\tau}|s_{\tau-1},x_{\tau-1})}{Q_{\theta}(s_{\tau}|s_{\tau-1},x_{\tau-1})}\right)\geq\sum_{(s,x)\in Gr(\Gamma)}E_{Q(\cdot\mid s,x)}\left[\ln\frac{Q(S'\mid s,x)}{Q_{\theta}(S'\mid s,x)}\right]m(s,x)-\frac{\alpha}{4}=K_{Q}(m,\theta)-\frac{1}{4}\alpha,\label{eq:SLLN-2}
\end{equation}
a.s.-$\mathbf{P}^{f}$ and$\forall\theta\in\Theta_{m,\xi}$. For any
$\theta\in\{\Theta\colon d_{m}(\theta)\geq\epsilon\}\cap\Theta_{m,\xi}$,
the RHS is bounded below by $K^{\ast}(m)+3\alpha-\frac{1}{4}\alpha>K^{\ast}(m)+\frac{3}{2}\alpha$.

Moreover, since $\lim_{t\rightarrow\infty}freq_{t}(h^{\infty},z)=\bar{P}_{m}(z)$
a.s.-$\mathbf{P}^{f}$ (uniformly over $z\in\mathbb{Z}$), there exists
a $T$ such that for any $\theta\notin\Theta_{m,\xi}$, $\kappa_{3t}(h^{\infty},\theta)\leq freq_{t}(z)\ln Q_{\theta}(z)\leq\left(p_{L}/2\right)\ln\xi$
for all $t\geq T(\xi)$ and some $z\in\mathbb{Z}$ where $p_{L}=\min_{\mathbb{Z}}\{\bar{P}_{m}(z):\bar{P}_{m}(z)>0\}$.
Therefore, for any $\theta\notin\Theta_{m,\xi}$ and a.s.-$\mathbf{P}^{f}$
: 
\begin{equation}
t^{-1}\sum_{\tau=1}^{t}\log\left(\frac{Q(s_{\tau}|s_{\tau-1},x_{\tau-1})}{Q_{\theta}(s_{\tau}|s_{\tau-1},x_{\tau-1})}\right)\geq\sum_{z\in\mathbb{Z}\colon Q(z)>0}\bar{P}_{m}(z)\ln Q(z)-\left(p_{L}/2\right)\ln\xi\label{eq:SLLN-3}
\end{equation}
for all $t\geq T(\xi)$. Observe that $\sum_{z\in\mathbb{Z}\colon Q(z)>0}\bar{P}_{m}(z)\ln Q(z)$
and $K^{\ast}(m)$ are bounded, so there exists a $\xi(\alpha)$ such
that the RHS can be made larger than $K^{\ast}(m)+\frac{3}{2}\alpha$. 

Therefore, by displays \ref{eq:SLLN-2} and \ref{eq:SLLN-3}, it follows
that: For any $t\geq T\equiv\max\{\hat{t}_{\zeta_{\xi(\alpha)}},T(\xi(\alpha))\}$
and a.s.-$\mathbf{P}^{f}$ 
\[
t^{-1}\sum_{\tau=1}^{t}\log\left(\frac{Q(s_{\tau}|s_{\tau-1},x_{\tau-1})}{Q_{\theta}(s_{\tau}|s_{\tau-1},x_{\tau-1})}\right)\geq K^{\ast}(m)+\frac{3}{2}\alpha
\]
for all $\theta\in\{\Theta\colon d_{m}(\theta)\geq\epsilon\}$, as
desired. $\square$

\bigskip{}

\textbf{Proof of Claim C.} We first show that for any $(s',x')\in Gr(\Gamma)$
and $(s_{0},x_{0})\in Gr(\Gamma)$, there exists an $n$ such that
$M_{\sigma,Q}^{n}(s',x'\mid s_{0},x_{0})>0$, where $M_{\sigma,Q}^{n}=M_{\sigma,Q}\cdots M_{\sigma,Q}$.\footnote{The expression $M_{\sigma,Q}\cdot M_{\sigma,Q}$ is defined as a transition
probability function over $\mathbb{S}\times\mathbb{X}$ where $M_{\sigma,Q}\cdot M_{\sigma,Q}(s',x'\mid s,x)\equiv\sum_{(a,b)}M(s',x'\mid a,b)M(a,b\mid s,x)$.
The expression $M_{\sigma,Q}\cdots M_{\sigma,Q}$ is constructed by
successive iterations of the previous one.} By the condition in Definition \ref{de:Fullcomm}, there exist an
$n$ and a ``path'' $((s_{1},x_{1}),...,(s_{n},x_{n}))$ such that
$(s_{i},x_{i})\in Gr(\Gamma)$ for all $i=1,...,n$ and 
\[
Q(s'\mid s_{n},x_{n})Q(s_{n}\mid s_{n-1},x_{n-1})...Q(s_{1}\mid s_{0},x_{0})>0.
\]
This inequality and the fact that $\sigma(x\mid s)\geq\varepsilon$
for all $(s,x)\in Gr(\Gamma)$, imply that 
\begin{align*}
M_{\sigma,Q}^{n}(s',x'\mid s_{0},x_{0}) & =\sum_{((s_{1},x_{1}),...,(s_{n},x_{n}))}\sigma(x'\mid s')Q(s'\mid s_{n},x_{n})...\sigma(x_{1}\mid s_{1})Q(s_{1}\mid s_{0},x_{0})\\
 & \geq\varepsilon^{n+1}\sum_{((s_{1},x_{1}),...,(s_{n},x_{n}))}Q(s'\mid s_{n},x_{n})...Q(s_{1}\mid s_{0},x_{0})>0,
\end{align*}
as desired. 

Consider any invariant distribution $m$. There exists at least one
point $(s_{0},x_{0})\in Gr(\Gamma)$ such that $m(s_{0},x_{0})>0$.
For any $(s',x')\in Gr(\Gamma)$, let $n$ be the integer that ensures
that $M_{\sigma,Q}^{n}(s',x'\mid s_{0},x_{0})>0$. Then, it follows
that $m(s',x')=\sum_{(s,x)\in Gr(\Gamma)}M_{\sigma,Q}^{n}(s',x'\mid s,x)m(s,x)\geq M_{\sigma,Q}^{n}(s',x'\mid s_{0},x_{0})m(s_{0},x_{0})>0$.
Thus, $supp(m)=Gr(\Gamma)$. $\square$

\subsection*{Computing $\Theta_{Q}(\cdot)$ and the stationary distribution in
the search example}

\textbf{Claim D.} \textbf{(i)} Let $\sigma$ be a strategy characterized
by a threshold $w^{\ast}$. Then there is a unique stationary marginal
distribution over $\mathbb{X}$, $m_{\mathbb{X}}(\cdot;w^{\ast})$,
and it is given by 
\[
m_{\mathbb{X}}(0;w^{\ast})=\frac{E[\gamma]-(1-F(w^{\ast}))E[\lambda\gamma]}{(1-F(w^{\ast}))\left\{ E[\lambda]-E[\lambda\gamma]\right\} +E[\gamma]}.
\]
 \textbf{(ii)} For any $m\in\Delta(Gr(\Gamma))$ with marginal $m_{\mathbb{X}}\in\Delta(\mathbb{X})$,
$\Theta_{Q}(m)$ is a singleton given by 
\begin{align*}
\theta_{Q}(m) & =\frac{m_{\mathbb{X}}(0)}{m_{\mathbb{X}}(0)+m_{\mathbb{X}}(1)\left(E\left[\gamma\right]\right)}\bar{\lambda}+\left(1-\frac{m_{\mathbb{X}}(0)}{m_{\mathbb{X}}(0)+m_{\mathbb{X}}(1)\bar{\gamma}}\right)\left(\bar{\lambda}+\frac{Cov(\gamma,\lambda)}{\bar{\gamma}}\right).
\end{align*}

\bigskip{}

\textbf{Proof of Claim D. }(i) For any $m\in\Delta(Gr(\Gamma))$,
$z',x'$ and $A\subseteq\mathbb{S}$ Borel, let 
\[
m(z',A,x')=\int_{\mathbb{S}}\int_{\mathbb{X}}\sigma(x'|w')\bar{Q}(z',A|w,x)m(w,x)dwdx,
\]
where $\{z'\},A,x'$ is just notation for the set $\{z'\}\times A\times\{x'\}$
and $\bar{Q}(z',A\mid w,x)\equiv\Pr(A|z',w,x)G(z')$, with 

\[
\Pr\left(w'\in A|z',w,0\right)=\left\{ \begin{array}{cc}
\int_{A}F(dw') & \text{ w/ prob. \ensuremath{\lambda(z')}}\\
1\{0\in A\} & \text{ w/ prob. }(1-\lambda(z'))
\end{array}\right.,
\]
and 
\[
\Pr\left(w'\in A|z',w,1\right)=\begin{cases}
\begin{array}{ccc}
\left\{ \begin{array}{cc}
\int_{A}F(dw') & \text{ w/pr. \ensuremath{\lambda(z')}}\\
1\{0\in A\} & \text{ w/pr. }(1-\lambda(z'))
\end{array}\right. & \text{ w/pr. }\gamma(z')\\
1\{w\in A\} & \text{ w/pr. }1-\gamma(z')
\end{array}\end{cases}.
\]

Also $\sigma(1|w)=1\{w>w^{\ast}\}$. Hence, for $x'=1$ 
\[
m(z',\mathbb{S},1;w^{\ast})=m(z',\{w'>w^{\ast}\},1;w^{\ast})
\]
and similarly for $x'=0$. Thus, $m_{\mathbb{X}}(1;w^{\ast})=\int_{\mathbb{Z}}m(dz',\{w'>w^{\ast}\},1;w^{\ast})$
and $m_{\mathbb{X}}(0;w^{\ast})=\int_{\mathbb{Z}}m(dz',\{w'<w^{\ast}\},0;w^{\ast})$
($w'=w^{\ast}$ occurs with probability zero, so it can be ignored).
It thus follows that 
\begin{align*}
m_{\mathbb{X}}(1;w^{\ast})= & \int_{\mathbb{Z}}\int_{\mathbb{S}}\int_{\mathbb{X}}\sigma(x'|w')\bar{Q}(dz',\{w'>w^{\ast}\}|w,x)m(w,x;w^{\ast})dwdx\\
= & \int_{\mathbb{Z}}\int_{\mathbb{S}}\int_{\mathbb{X}}\Pr\left(\{w'>w^{\ast}\}|z',w,x\right)G(dz')m(w,x;w^{\ast})dwdx\\
= & \int_{\mathbb{Z}}\int_{\mathbb{S}}\Pr\left(\{w'>w^{\ast}\}|z',0\right)G(dz')m(w,0;w^{\ast})dw\\
 & +\int_{\mathbb{Z}}\int_{\mathbb{S}}\Pr\left(\{w'>w^{\ast}\}|z',w,1\right)G(dz')m(w,1;w^{\ast})dw\\
= & \int_{\mathbb{Z}}\lambda(z')G(dz')(1-F(w^{\ast}))m_{\mathbb{X}}(0;w^{\ast})+\int_{\mathbb{Z}}\gamma(z')\lambda(z')(1-F(w^{\ast}))G(dz')m_{\mathbb{X}}(1;w^{\ast})\\
 & +\int_{\mathbb{Z}}(1-\gamma(z'))G(dz')\int_{\mathbb{S}}1\{w>w^{\ast}\}m(dw,1;w^{\ast}).
\end{align*}
where the last line follows from the fact that $1\{w'>w^{\ast}\}1\{w'=0\}=0$
always. Observe that $\int_{\mathbb{W}}1\{w>w^{\ast}\}m(dw,1;w^{\ast})=m(\{w>w^{\ast}\},1;w^{\ast})=m_{\mathbb{X}}(1;w^{\ast})$
by our previous observation. Thus 
\[
m_{\mathbb{X}}(1;w^{\ast})=E[\lambda](1-F(w^{\ast}))m_{\mathbb{X}}(0;w^{\ast})+\left\{ E[\lambda\gamma](1-F(w^{\ast}))+(1-E[\gamma])\right\} m_{\mathbb{X}}(1;w^{\ast}).
\]
Solving for $m_{\mathbb{X}}(1;w^{\ast})$, we obtain 
\[
m_{\mathbb{X}}(1;w^{\ast})=\frac{E[\lambda](1-F(w^{\ast}))}{(1-F(w^{\ast}))\left\{ E[\lambda]-E[\lambda\gamma]\right\} +E[\gamma]}.
\]
The result follows by noting that $m_{\mathbb{X}}(0;w^{\ast})=1-m_{\mathbb{X}}(1;w^{\ast})$.

(ii) For $x=1$, 
\begin{align*}
E_{Q(\cdot|w,1)}\left[\ln\left(\frac{Q(W'|w,1)}{Q_{\theta}(W'|w,1)}\right)\right]= & \sum_{z'\in\mathbb{Z}}\gamma(z')\left\{ \lambda(z')\int_{0}^{1}\log\left(\theta\right)F(dw')+(1-\lambda(z'))\log\left(1-\theta\right)\right\} G(z')\\
 & +\sum_{z'\in\mathbb{Z}}(1-\gamma(z'))\left\{ \log\left(1\right)\right\} G(z')\\
= & E[\lambda\gamma]\log\left(\theta\right)+(E[\gamma]-E[\lambda\gamma])\log\left(1-\theta\right)+Const.
\end{align*}
Similarly, for $x=0$, 
\[
E_{Q(\cdot|w,0)}\left[\ln\left(\frac{Q(W'|s,0)}{Q_{\theta}(W'|s,0)}\right)\right]=E[\lambda]\log\left(\theta\right)+(1-E[\lambda])\log\left(1-\theta\right)+Const',
\]
where $Const$ and $Const'$ are constants that do not depend on $\theta$.
It is easy to see that, over $[0,1]$, these are strictly convex functions
of $\theta$, so a convex combination also is. Thus $\Theta_{Q}(m)$
is a singleton for any $m$, which we denote as $\theta_{Q}(m)$.
The first order conditions yield 
\[
\frac{1}{\theta}\left\{ E[\lambda\gamma]m_{\mathbb{X}}(1)+E[\lambda]m_{\mathbb{X}}(0)\right\} =\frac{1}{1-\theta}\left\{ \left(E[\gamma]-E[\lambda\gamma]\right)m_{\mathbb{X}}(1)+(1-E[\lambda])m_{\mathbb{X}}(0)\right\} .
\]
Thus 
\[
\theta_{Q}(m)=\frac{E[\lambda\gamma]m_{\mathbb{X}}(1)+\bar{\lambda}m_{\mathbb{X}}(0)}{\bar{\gamma}m_{\mathbb{X}}(1)+m_{\mathbb{X}}(0)}.
\]
The desired results follows from some algebra and from the standard
expression for the covariance. $\square$

\end{document}